\documentclass[journal,onecolumn,draftcls]{IEEEtran}

\ifCLASSINFOpdf
\else
\fi
\usepackage{amsmath,amssymb,amsfonts}
\usepackage{algorithmic}
\usepackage{graphicx}
\usepackage{textcomp}
\usepackage{enumerate}
\usepackage{hyperref}
\usepackage{tabularx,booktabs}
\usepackage{amsthm}
\usepackage{bbm}
\usepackage{csquotes}
\usepackage{chngcntr}
\usepackage{apptools}
\usepackage{mathtools}
\usepackage{subcaption}
\usepackage{tikz}
\usepackage{caption} 
\usepackage{cleveref}
\usepackage{upgreek}
\usetikzlibrary{mindmap,backgrounds,arrows.meta}

\usepackage{xcolor}

\newcommand{\Pm}{\mathcal{P}}

\newcommand{\Q}{\mathcal{Q}}
\newcommand{\X}{\mathcal{X}}

\newcommand{\F}{\mathcal{F}}

\DeclareMathOperator*{\esssup}{ess\,sup}

\DeclareMathOperator{\sign}{sign}
\theoremstyle{plain}
\newtheorem{theorem}{Theorem}
\newtheorem{lemma}{Lemma}
\newtheorem{corollary}{Corollary}
\newtheorem{proposition}{Proposition}
\newtheorem*{lemma*}{Lemma}

\theoremstyle{definition}
\newtheorem{definition}{Definition}

\theoremstyle{remark}

\newtheorem{remark}{Remark}
\newtheorem{example}{Example}
\makeatletter
\newcounter{labelcnt}
\renewcommand{\thelabelcnt}{(\alph{labelcnt})}
\newcommand{\setlabel}[1]{%
	\refstepcounter{labelcnt}\ltx@label{lbl:#1}%
	{\text{\upshape\thelabelcnt}}%
}

\makeatother
\usepackage[sort&compress,square,numbers]{natbib}

\DeclareMathOperator{\dd}{d\!}

\newcommand{\argmax}{\operatornamewithlimits{argmax}}
\newcommand{\argmin}{\operatornamewithlimits{argmin}}

\begin{document}
	%

	\title{Concentration without Independence \\via Information Measures}
	%
	%
	%
	
	\author{
     \IEEEauthorblockN{Amedeo Roberto Esposito, Marco Mondelli}\\\IEEEauthorblockA{Institute of Science and Technology Austria
    \\\{amedeoroberto.esposito, marco.mondelli\}@ist.ac.at}}

	\maketitle
	
	\begin{abstract}
	 We propose a novel approach to concentration for non-independent random variables. 
  The main idea is to ``pretend'' that the random variables are independent and pay a multiplicative price measuring how far they are from actually being independent. This price is encapsulated in the Hellinger integral between the joint and the product of the marginals, which is then upper bounded leveraging tensorisation properties. Our bounds represent a natural generalisation of concentration inequalities in the presence of dependence: we recover exactly the classical bounds (McDiarmid's inequality) when the random variables are independent. Furthermore, in a ``large deviations'' regime, we obtain the same decay in the probability as for the independent case, even when the random variables display non-trivial dependencies.
  To show this, we consider a number of applications of interest. First, we provide a bound for Markov chains with finite state space. Then, we consider the Simple Symmetric Random Walk, which is a non-contracting Markov chain, and a non-Markovian setting in which the stochastic process depends on its entire past. To conclude, we propose an application to Markov Chain Monte Carlo methods, where our approach leads to an improved lower bound on the minimum burn-in period required to reach a certain accuracy. In all of these settings, we provide a regime of parameters in which our bound fares better than what the state of the art can provide.
	\end{abstract}
	

	%
	\IEEEpeerreviewmaketitle

	\section{Introduction}
It is well-known that, given a sequence $X^n=(X_1,\ldots, X_n)$ of independent, but not necessarily identically distributed, random variables with joint measure $\Pm_{X^n}$, one can prove that for every function $f$ satisfying proper Lipschitz assumptions:
	\begin{equation}
	    \Pm_{X^n}(|f-\Pm_{X^n}(f)|\geq t) \leq 2\exp\left(-\frac{t^2}{k \left\lVert f\right\rVert_{\text{Lip}}^2}\right).
	\end{equation}
	Here, $\left\lVert f\right\rVert_{\text{Lip}}^2$ depends on the metric structure of the measure space, and $k$ is a constant depending on the approach used to prove the inequality, \emph{e.g.}, transportation-cost inequalities, log-Sobolev inequalities, martingale method, see the survey~\cite{concentrationMeasureII}. 
	One notable example is McDiarmid's inequality for functions with ``bounded jumps'': \textit{i.e.}, if for every $x^n,\hat{x}$ and every $1\leq i \leq n$ one has that \begin{equation}|f(x_1,\ldots,x_i,\ldots,x_n)-f(x_1,\ldots,\hat{x},\ldots,x_n)|\leq 
c_i \label{eq:mcDiarmidAssumption}, \end{equation} then the following holds \cite{mcdiarmid_1989}:
	\begin{equation}
	    \Pm_{X^n}(|f-\Pm_{X^n}(f)|\geq t) \leq 2\exp\left(-\frac{2t^2}{\sum_{i=1}^n c^2_i}\right). \label{eq:mcDiarmidStatement}
	\end{equation}
	This represents the ``golden standard'' of concentration. Interestingly, as underlined above, McDiarmid's inequality does not require the $X_i$'s to be identically distributed; it does, however, require the random variables to be \emph{independent}.
Most of the methods in the literature that tried to relax the latter assumption 
 required the development of novel techniques. However, existing results generally do not recover the rate of decay provided in the independent setting.
 
 In this paper, we present a novel approach that 
 outperforms the state of the art in various settings and regimes. Specifically, we show improved bounds for 
 finite-state space Markov chains (\Cref{sec:example1}), the Simple Symmetric Random Walk (SSRW, \Cref{sec:example2}), a non-Markovian process (\Cref{sec:example3}), and Monte Carlo Markov Chain (MCMC, \Cref{sec:MCMC}).
  In the case of the SSRW, our improvements are the most dramatic: in sharp contrast with existing techniques, we are able to capture the correct scaling between the distance from the average $t$, the number of variables $n$, and the decay probability in the concentration bound. 
We remark that our new method -- based on a change of measure argument -- is rather flexible and can be employed in most settings. In fact, it only requires the absolute continuity between the joint and the product of the marginals, while existing approaches generally have more restrictive assumptions (\textit{e.g.}, Markovianity with stationary distribution~\cite{dependentViaStatAndChange} or contractivity~\cite{Marton1996,Marton1996BoundingB}). The key idea is to shift the focus 
from proving concentration to bounding an information measure (\textit{i.e.}, the Hellinger integral, see~\Cref{def:hellInt}) between the joint distribution and the product of the marginals. Crucially, the Hellinger integral satisfies tensorisation properties that allow us to easily upper bound it, even in high-dimensional settings (see~\Cref{app:tensorisation}).
We highlight that our approach provides a \emph{natural} generalisation of the existing concentration of measure results to dependent random variables, in the sense that we recover exactly McDiarmid's inequality when the random variables are independent. Furthermore, for sufficiently large $t$, namely, in a ``large deviations'' regime, we approach the decay rate \eqref{eq:mcDiarmidStatement} for the independent case, even when the random variables are actually dependent.

The rest of the paper is organized as follows. In~\Cref{ss:relatedWork}, we discuss related work in the area.
in~\Cref{sec:preliminaries} we cover the preliminaries, namely, information measures (\Cref{sec:infoMeas}), Markov kernels (\Cref{sec:markovPreliminaries}), and strong data-processing inequalities (SDPIs, \Cref{sec:sdpi}). We then provide the main result of this work in~\Cref{sec:mainResult}, which is then applied in~\Cref{sec:examples} to four different settings: finite-state space Markov chains (\Cref{sec:example1}), the SSRW (\Cref{sec:example2}) a non-Markovian stochastic process (\Cref{sec:example3}), and MCMC methods (\Cref{sec:MCMC}). Concluding remarks are provided in~\Cref{sec:conclusions}. Part of the proofs and additional discussions are deferred to the appendices. 

\subsection{Related Work}\label{ss:relatedWork}
 The problem of concentration for dependent random variables has been addressed in multiple ways. The first results in the area are due to Marton~\cite{Marton1996, Marton1996BoundingB, Marton1998, Marton2003} who heavily relied on transportation-cost inequalities (Pinsker-like inequalities) and an elegant mixture of information-theoretic and geometric approaches. Another important contribution, building upon Marton's work, was given by Samson in~\cite{samsonConcentration} where some of Marton's results were extended to include $\Phi$-mixing processes. More recent advances, complementing and generalising the work by Samson and Marton, were provided in~\cite{dependentViaMartingale}, where the Martingale method was employed to prove concentration for dependent (but defined on a countable space) processes and in~\cite{dependentViaCoupling, dependentViaMartonCouplingAndSpectral}, where the idea of couplings was exploited. In particular, the results derived in~\cite{dependentViaCoupling} are equivalent to the ones advanced in~\cite{dependentViaMartingale} but obtained through couplings rather than linear programming. Moreover, ~\cite{dependentViaMartonCouplingAndSpectral} leverages Marton's coupling. For an extensive treatise on Marton's coupling please refer to~\cite[Chapter 8]{BLM2013Concentration}. All of these approaches measure the degree of dependence by looking at distances between conditional distributions (organised in matrices whose norms are then computed, see~\Cref{eq:totalVarationDependentCase}) or by constructing ``minimal couplings'' between conditional distributions (see~\Cref{eq:martonDependentCase}). The resulting quantities, which are necessary in order to 
 analyse the corresponding probabilities, can be difficult to compute, especially in non-Markovian settings. 
 Another approach similar in spirit to ours is given in~\cite{dependentViaStatAndChange}, where a generalisation of Hoeffding's inequality for stationary Markov chains is provided. Other related work can be found in~\cite{markovChain1,markovChain2,markovChain3,markovChain4,markovChain5}. These results aim to establish Hoeffding-like inequalities for Markov chains by relating it to a different Markov chain whose cumulant generating function can be bounded under different assumptions:~\cite{markovChain1,markovChain2,markovChain3,markovChain4} are restricted to discrete and ergodic Markov chains, while~\cite{markovChain5} extends to general state-space but requires geometric ergodicity. Yet another approach in providing exponential concentration for geometrically ergodic Markov chains can be found in~\cite{subgaussianErgodicMarkovChains, mcDiarmidGeometricMarkovChains}. All these generalisations of Hoeffding's inequality do not, however, allow for arbitrary functions of a sequence of random variables, but they are restricted to (sums of) bounded functions applied to each individual sample and they all require the existence of a stationary distribution.  Given that the approach presented in~\cite{dependentViaStatAndChange} is more general than the one proposed in~\cite{markovChain1,markovChain2,markovChain3,markovChain4,markovChain5}, our results will be compared directly with \cite{dependentViaStatAndChange}.
    Another (less related) approach can be found in~\cite{dependentAzimin}, where the strength of the dependence is measured in a different way with respect to both this work and the related work mentioned above. Moreover, the approach in~\cite{dependentAzimin} is mostly restricted to empirical averages of bounded random variables and includes an additional additive factor that grows with the number of samples. Exponential bounds for stochastic chains of unbounded memory on countable alphabets are instead given in~\cite{gaussianBoundsStochasticChains}. Additionally, an approach that leverages the Kullback-Leibler divergence can be found in~\cite{KLMC,MCMCKL}.
    Finally, we remark that 	\cite{fullVersionGeneralization} exploits a technique similar to what is pursued in this work, in order to extend McDiarmid's inequality to the case where the function $f$ depends on the random variables themselves (while the random variables remain, in fact, independent). Said result was then applied to a learning setting.

 \section{Preliminaries}\label{sec:preliminaries}

 In this section, we will define the main objects utilised throughout the document and define the relevant notation.
	We will 
 adopt a measure-theoretic framework.
	Given a measurable space $(\X,\mathcal{F})$ and two measures $\mu,\nu$ which render it a measure space, if $\nu$ is absolutely continuous with respect to $\mu$ (denoted with $\nu\ll\mu$), then we will represent with $\frac{d\nu}{d\mu}$ the Radon-Nikodym derivative of $\nu$ with respect to $\mu$. Given a (measurable) function $f:\X \to \mathbb{R}$ and a measure $\mu$, we denote with $\mu(f) = \int f \dd \mu$ the 
Lebesgue integral of $f$ with respect to the measure $\mu$. 
The Radon-Nikodym derivatives represent the main building block of the following fundamental objects.

	\subsection{Hellinger integral, $\alpha$-norm and R\'enyi's $\alpha$-divergence}\label{sec:infoMeas}
    An important ingredient of this work is information measures. In particular, we will focus on Hellinger integrals which can be seen as a transformation of the $L^\alpha$-norms and of
 R\'enyi's $\alpha$-divergences. Let us introduce them and then show the relationships with other well-known objects in the literature. Hellinger integrals can be seen as a $\varphi$-Divergence with a specific parametrised choice of $\varphi$~\cite{fDiv1}. 
 	\begin{definition}[$\varphi$-divergences]\label{def:fDiv}
		Let $(\Omega,\F,\Pm),(\Omega,\F,\Q)$ be two probability spaces. Let $\varphi:\mathbb{R}^+\to \mathbb{R}$ be a convex function such that $\varphi(1)=0$. Consider a measure $\mu$ such that $\Pm\ll\mu$ and $\Q\ll\mu$. Denoting with $p,q$ the densities of the measures with respect to $\mu$, the \emph{$\varphi$-divergence of $\Pm$ from $\Q$} is defined as 
		\begin{align}
			D_\varphi(\Pm\|\Q):=\int q \varphi\left(\frac{p}{q}\right) \dd\mu.
		\end{align}
	\end{definition} 
Particularly relevant to us will be the family of parametrised divergences that stems from $\varphi_\alpha(x)=x^\alpha$ for $\alpha>1$. The function $\varphi_\alpha(x)$ is convex on the positive axis for every $\alpha>1$. However, it does not satisfy the property that $\varphi(1)=0$. Said requirement can be lifted with the consequence of losing the property that $D_\varphi(\nu\|\mu)=0$ if and only if $\nu = \mu$. We will call the family of divergences stemming from such functions the \emph{Hellinger integrals} of order $\alpha$.
	\begin{definition}[Hellinger integrals]
		Let $(\Omega,\F,\nu),(\Omega,\F,\mu)$ be two probability spaces, and let 
  $\varphi_\alpha:\mathbb{R}^+\to \mathbb{R}$ be defined as $\varphi_\alpha(x)=x^\alpha$. Let $\mu$ and $\nu$ be two probability measures such that $\nu\ll\mu$, then the \emph{Hellinger integral} of order $\alpha$ is given by
		\begin{align}
			H_\alpha(\nu\|\mu):= D_{\varphi_\alpha}(\nu\|\mu) = \int \left(\frac{d\nu}{d\mu}\right)^\alpha \dd\mu.
		\end{align}
	\label{def:hellInt}\end{definition}
Let us highlight that we are not considering Hellinger divergences of order $\alpha$ (including the so-called $\chi^2$-divergence) which consist of divergences stemming from $\frac{x^\alpha-1}{(\alpha-1)}$, but rather a transformation of said family. In fact, the Hellinger divergences are equal to $0$ if and only if the measures coincide. In contrast, the Hellinger integral is equal to $1$ if the two measures coincide.
 \begin{remark}[$\varphi$-Divergences]
	Despite the fact that~\Cref{def:fDiv} uses a reference measure $\mu$ and the densities with respect to this measure, 
 $\varphi$-divergences can be shown to be 
 independent from the dominating measure. In fact, when absolute continuity between $\Pm,\Q$ holds, \emph{i.e.}, $\Pm\ll\Q$,\footnote{We will make this assumption throughout the paper.}
 we obtain~\cite{fDiv1}
	\begin{equation}
		D_\varphi(\Pm\|\Q)= \int \varphi\left(\frac{d\Pm}{d\Q}\right)\dd\Q.
	\end{equation}
Moreover, $\varphi$-Divergences can be seen as a generalisation of well-known objects like the Kullback-Leibler Divergence. Indeed,
the KL-divergence is retrieved by setting $\varphi(t)=t\log(t)$. Other common examples are the Total Variation distance ($\varphi(t)=\frac12|t-1|$), the Hellinger distance ($\varphi(t)=(\sqrt{t}-1)^2$), and Pearson $\chi^2$-divergence ($\varphi(t)=t^2-1$). We remark that $\varphi$-divergences do not include the family of R\'enyi's $\alpha$-divergences.
 \end{remark}

Hellinger integrals, other than belonging to the family of $\varphi$-Divergences can also be related to other notable objects, namely: R\'enyi Divergences of order $\alpha$ and $L^\alpha$-norms~\cite{renyiEntropy,RenyiKLDiv}:
	\begin{definition}[R\'enyi divergences]\label{def:Renyi}
		Let $(\Omega,\F,\Pm),(\Omega,\F,\Q)$ be two probability spaces. Let $\alpha>0$ be a positive real number different from $1$. Consider a measure $\mu$ such that $\Pm\ll\mu$ and $\Q\ll\mu$ (such a measure always exists, \emph{e.g.}, $\mu=(\Pm+\Q)/2$)) and denote with $p,q$ the densities of $\Pm,\Q$ with respect to $\mu$. Then, the \emph{$\alpha$-divergence of $\Pm$ from $\Q$} is defined as 
		\begin{align}
			D_\alpha(\Pm\|\Q):=\frac{1}{\alpha-1} \log \int p^\alpha q^{1-\alpha} \dd\mu.
		\end{align}
	\end{definition}
     \begin{remark}
         	Definition \ref{def:Renyi} is independent of the chosen measure $\mu$.
		In fact, 
		$\int p^{\alpha}q^{1-\alpha} \dd\mu = \int \left(\frac{q}{p}\right)^{1-\alpha}\dd\Pm $ and, whenever $\Pm\ll\Q$ or $0<\alpha<1$, we have $\int p^{\alpha}q^{1-\alpha} \dd\mu= \int \left(\frac{p}{q}\right)^{\alpha}\dd\Q$, see \cite{RenyiKLDiv}. Furthermore, it can be shown that, if $\alpha>1$ and $\Pm\not\ll\Q$, then $D_\alpha(\Pm\|\Q)=\infty$. The behavior of the measure for $\alpha\in\{0,1,\infty\}$ can be defined by continuity. These objects can also be seen as a generalisation of the Kullback-Leibler Divergence. Indeed, one has that $D_1(\Pm\|\Q) = D(\Pm\|\Q)$ which denotes the KL-divergence between $\Pm$ and $\Q$; furthermore, if $D(\Pm\|\Q)=\infty$ or there exists $\beta>1$ such that $D_\beta(\Pm\|\Q)<\infty$, then $\lim_{\alpha\downarrow1}D_\alpha(\Pm\|Q)=D(\Pm\|\Q)$~\cite[Theorem 5]{RenyiKLDiv}. For an extensive treatment of $\alpha$-divergences and their properties, we refer the reader to~\cite{RenyiKLDiv}. 
     \end{remark}
Going back to the Hellinger integral, the following relationship holds:
	\begin{equation}
	    H_\alpha(\nu\|\mu)= \left\lVert \frac{d\nu}{d\mu} \right\rVert_{L^\alpha(\mu)}^\alpha = \exp\left((\alpha-1)D_\alpha(\nu\|\mu)\right) \label{eq:hellignerRenyi},
	\end{equation}
	where $\left\lVert \frac{d\nu}{d\mu} \right\rVert_{L^\alpha(\mu)}$ denotes the $L^\alpha$-norm of the Radon-Nikodym derivative with respect to the measure $\mu$.
    \subsection{Markov kernels}\label{sec:markovPreliminaries}
    Most of the comparisons with the state of the art will be drawn in 
    Markovian settings. In this section, we will define the main objects necessary in order to carry out said confrontation.
    \begin{definition}[Markov kernel]
    Let $(\Omega,\mathcal{F})$ be a measurable space. A \emph{Markov kernel} $K$ is a mapping $K:\mathcal{F}\times \Omega \to [0,1]$ such that:
    \begin{enumerate}
        \item for every $x\in \Omega$, the mapping $E\in\mathcal{F}\to K(E|x)$ is a probability measure on $(\Omega,\mathcal{F})$;
        \item for every $E\in\mathcal{F}$ the mapping $x \in \Omega \to K(E|x)$ is an $\mathcal{F}$-measurable real-valued function.
    \end{enumerate}
    \end{definition}
    A Markov kernel can be seen as acting on measures ``from the right'', \textit{i.e.}, given a measure $\mu$ on $(\Omega, \mathcal{F})$, 
    \begin{equation}
    \mu K(E) = \mu(K(E|\cdot)) = \int \dd\mu(x) K(E|x),
    \end{equation}
    and on functions ``from the left'',  \textit{i.e.}, given a function $f:\Omega \to \mathbb{R}$,
    \begin{equation}
        K f(x) = \int dK(y|x) f(y).
    \end{equation}
    Given a sequence of random variables $(X_n)_{n\in\mathbb{N}}$, one says that it represents a \emph{Markov chain} if, given $i\geq 1$,  there exists a Markov kernel $K_i$ such that for every measurable event $E$:
    \begin{equation}
        \mathbb{P}(X_i \in E | X_1,\ldots,X_{i-1}) = \mathbb{P}(X_i \in E | X_{i-1}) = K_i(E|X_{i-1})\qquad  \text{ almost surely}.
    \end{equation}
    If for every $i\geq 1$, $K_i = K$ for some Markov kernel $K$, then the Markov chain is said to be time-homogeneous. Whenever the index is suppressed from $K$, we will be referring to a time-homogeneous Markov chain. The kernel $K$ of a Markov chain describes the probability of getting from $x$ to $E$ in one step, \textit{i.e.}, for every $i\geq 1$, 
    $K(E|x) = \mathbb{P}(X_i \in E| X_{i-1}=x)$. One can then define (inductively) the $\kappa$-step 
    kernel $K^\kappa$ as follows:
    \begin{equation}
        K^\kappa(E|x) = \int K^{\kappa-1}(E|y)dK(y|x).
    \end{equation}
    Note that $K^\kappa$ is also a Markov kernel, and it represents the probability of getting from $x$ to $E$ in $\kappa$ steps: $K^\kappa(E|x)= \mathbb{P}(X_{\kappa+1}\in E|X_{1}=x)$.
    If $(X_n)_{n\in\mathbb{N}}$ is the Markov chain associated to the kernel $K$ and $X_0 \sim \mu$, then $\mu K^m$ denotes the measure of $X_{m+1}$ at every $m\in \mathbb{N}$. Furthermore, a probability measure $\pi$ is a stationary measure for $K$ if $\pi K(E) = \pi(E)$ for every measurable event $E$. We also note that, if the state space is discrete, then $K$ can be represented using a stochastic matrix. 
    
    Given this dual perspective on Markov operators (acting on measures or functions), one can then study their contractive properties. In particular, let us define 
    \begin{equation}
            \left\lVert K\right\rVert_{\alpha\to\alpha} := \sup_{f\neq 0} \frac{\left\lVert K f\right\rVert_{\alpha}}{\left\lVert f \right\rVert_{\alpha}}.\label{eq:contractivityCoefficient}
        \end{equation}
    Then, Markov kernels are generally contractive~\cite{spreadingOfSetsHypercontractivity}, meaning that $ \left\lVert K\right\rVert_{\alpha\to\alpha} \leq 1$ and, consequently, $
        \left\lVert Kf \right\rVert_{\alpha} \leq \left\lVert f \right\rVert_\alpha$ for every $f$.
        Similarly, given $\gamma\leq \alpha$, one can define the following quantity \begin{equation}
            \left\lVert K\right\rVert_{\alpha\to \gamma} := \sup_{f\neq 0} \frac{\left\lVert K f\right\rVert_{\alpha}}{\left\lVert f \right\rVert_{\gamma}}.\label{eq:hypercontractivity}
        \end{equation}
        It has been proven that many Markovian operators are hyper-contractive~\cite{spreadingOfSetsHypercontractivity,hypercontractivityBonamiBeckner,hypercontractivity}, meaning that $\left\lVert K\right\rVert_{\alpha\to \gamma}\leq 1$
        for some $\gamma < \alpha$.  Given a kernel $K$ and $\alpha>1$, we denote by $\gamma^\star_K(\alpha)$ the smallest $\gamma$ such that $K$ is hyper-contractive, \emph{i.e.}, such that  $\left\lVert K\right\rVert_{\alpha\to \gamma}\leq 1$. Said coefficient has been characterised for some Markov operators~\cite{hypercontractivity,spreadingOfSetsHypercontractivity}. In case the Markov kernel is not time-homogeneous, in order to simplify the notation, instead of denoting the corresponding coefficient with $\gamma^\star_{K_i}(\alpha)$, we will simply denote it with $\gamma^\star_{i}(\alpha)$. 
        
        Given a Markov kernel $K$ and a measure $\mu$, one can also define the adjoint/dual operator (or backward channel) $K^\leftarrow$ as the operator such that
         $\langle g, Kf\rangle = \langle K^\leftarrow g, f\rangle$ for all $g$ and $f$~\cite[Eq. (1.1)]{sdpiRaginsky}. While one can define dual Markovian operators more generally, here we will focus on discrete settings where they can be explicitly specified via $K$ and $\mu$~\cite[Eq. (1.2)]{sdpiRaginsky}:
    \begin{equation}
            K_\mu^\leftarrow(y|x) = \frac{K(y|x)\mu(x)}{\mu K (y)}.
        \end{equation}
     \subsection{Strong Data-Processing Inequalities}\label{sec:sdpi}
    An important property shared by divergences is the Data-Processing Inequality (DPI): given two measures $\mu,\nu$ and a Markov kernel $K$, one has that, for every convex $\varphi$,
    \begin{equation}
        D_\varphi(\nu K\| \mu K) \leq D_\varphi(\nu\|\mu).\label{eq:DPI}
    \end{equation}
    This property holds as well for R\'enyi's $\alpha$-divergences, despite them not being a $\varphi$-divergence~\cite[Theorem 9]{RenyiKLDiv}.
    DPIs represent a widely used tool and a line of work has focused on tightening them. In particular, in many settings of interest, given a reference measure $\mu$, one can show that $D_\varphi(\nu K\|\mu K)$ is strictly smaller than $D_\varphi(\nu\|\mu)$  unless $\nu=\mu$. Furthermore, the characterization of the ratio $D_\varphi(\nu K\|\mu K)/D_\varphi(\nu\|\mu)$ has lead to the study of ``strong Data-Processing Inequalities''~\cite[Definition 3.1]{sdpiRaginsky}.  
    
\begin{definition}[Strong Data-Processing Inequalities]
    Given a probability measure $\mu$, a Markov kernel $K$ and  a convex function $\varphi$, we say that $K$ satisfies a $\varphi$-type \emph{Strong Data-Processing Inequality (SDPI)} at $\mu$ with constant $c\in[0,1)$ if
    \begin{equation}
    D_\varphi(\nu K\|\mu K) \leq c\cdot D_\varphi(\nu\|\mu),
    \end{equation}
    for all $\nu\ll\mu$. The tightest such constant $c$ is denoted by
\begin{align*}
\eta_\varphi(\mu,K) &= \sup_{\nu\neq \mu} \frac{D_\varphi(\nu K\|\mu K) }{D_\varphi(\nu\|\mu)}, \\
\eta_\varphi(K) &= \sup_{\mu} \eta_\varphi(\mu,K).
\end{align*}
\end{definition}
 \begin{example}[SDPI for the KL and the BSC]
     Let $\mu=\text{Ber}(1/2)$, $\epsilon<\frac12$ and $K=\text{BSC}(\epsilon)$, \textit{i.e.}, $K(y|x) =\epsilon$ if $x=y$ and $K(y|x) =1-\epsilon$ otherwise. Then, one has that $\eta_{x\log x}(\mu,K)=(1-2\epsilon)^2$~\cite{spreadingOfSetsHypercontractivity}, which implies that  $\eta_{x\log x}(\mu,K) < 1$ for all $\epsilon>0$. 
 \end{example} 
 While $\eta_\varphi$ can be a difficult object to compute even for simple channels, some universal upper and lower bounds are known~\cite[Theorems 3.1, 3.3]{sdpiRaginsky}:
 \begin{equation}
 \eta_\varphi(K) \leq \sup_{x,\hat{x}} \left\lVert K(\cdot|x)-K(\cdot|\hat{x})\right\rVert_{TV} = \eta_{|x-1|}(K)=\eta_{TV}(K), \label{eq:etaTV}
 \end{equation}
  \begin{align}
 \eta_\varphi(\mu,K) \geq \eta_{(x-1)^2}(\mu,K) = \eta_{\chi^2}(\mu,K).\end{align}
 We remark that these bounds hold for functions $\varphi$ such that $\varphi(1)=0$ or, equivalently, when the divergence $D_\varphi(\nu\|\mu)$ is defined to be $\mu\left(\varphi\left(\frac{d\nu}{d\mu}\right)\right)-\varphi(1)$. For general convex functions $\varphi$, as well as for R\'enyi's divergences, the DPI holds and SDPI constants are still defined analogously, however one cannot use common techniques to bound said quantities. The following counter-example highlights the issue.
 \begin{example}[Counter-example for Hellinger integrals and R\'enyi's divergences] \label{ex:counterExDalpha}
     Let $\nu=(1/3,2/3)$ and $K_1=\text{BSC}(1/3)$. Then, the stationary distribution $\pi$ is given by $(1/2,1/2)$ and $\nu K_1= (5/9,4/9).$ A direct calculation gives that $H_2(\nu K_1\|\pi K_1) = \frac{82}{81}$ and $H_2(\nu\|\pi) = \frac{10}{9}$. Moreover, if $K=\text{BSC}(\lambda)$, one has that $\eta_{TV}(K)=|1-2\lambda|$ (see~\cite[Remark 3.1]{sdpiRaginsky}  and~\Cref{eq:etaTV}). Thus,  
        \begin{equation}
        \frac{H_2(\nu K_1\|\pi K_1)}{H_2(\nu\|\pi)} = \frac{41}{45} > \eta_{TV}(K_1) = \frac13, \end{equation}
        which means that the inequality \eqref{eq:etaTV} is violated. This is due to the fact that $\varphi_2(x)=x^2$ is not equal to $0$ at $x=1$. In fact, 
        renormalising $H_2$ leads to the $\chi^2$-divergence, which satisfies \begin{equation}
        \frac{H_2(\nu K_1\|\pi K_1)-1}{H_2(\nu\|\pi)-1}=\frac{\chi^2(\nu K_1\|\pi K_1)}{\chi^2(\nu \|\pi )}=\frac19 < \eta_{TV}(K_1) = \frac13.\end{equation}
        
  Similarly, let $K_2=\text{BSC}(1/5)$, which gives that $\eta_{TV}(K_2)=3/5$. Consider now $D_\alpha(K_2(\cdot|0)\|\pi)=D_\alpha(\delta_0 K_2 \|\pi K) = \frac{1}{\alpha-1} \log ( 2^{1-\alpha}( 0.2^\alpha +(0.8)^\alpha))$. Moreover, $D_\alpha(\delta_0\| \pi) = \log(2)$. Thus, by setting $\alpha=6$, one has that
     \begin{equation}
         \eta_{D_\alpha}(K_2) > \frac{D_\alpha(\delta_0 K_2 \|\pi K_2)}{D_\alpha(\delta_0\| \pi)} = 0.6138 > \eta_{TV}(K_2) = 0.6,
     \end{equation}
     which violates again the inequality \eqref{eq:etaTV}.
 \end{example}
    \section{Main result}\label{sec:mainResult}
\begin{theorem}\label{thm:main}
    Let $\Pm_{X^n}$ be the joint distribution of $(X_1,\ldots, X_n)$, $\Pm_{X_i}$ the marginal corresponding to $X_i$, and $\Pm_{\bigotimes_{i=1}^n X_i}$ the joint measure induced by the product of the marginals.     
  If $\Pm_{X^n}\ll\Pm_{\bigotimes_{i=1}^n X_i}$, for any function $f$ satisfying~\Cref{eq:mcDiarmidAssumption}, any $t>0$ and 
  $\alpha>1$, 
  one has
    \begin{align}           
        \Pm_{X^n}\left(\left\lvert f-\Pm_{\bigotimes_{i=1}^n X_i}(f)\right\rvert \geq t \right)&\leq 2^{\frac1\beta} \exp\left(\frac{-2t^2}{\beta \sum_{i=1}^n c_i^2}\right)H_\alpha^\frac1\alpha(\Pm_{X^n}\|\Pm_{\bigotimes_{i=1}^n X_i}) .\label{eq:generalTheorem}
    \end{align}
    Moreover, one can further upper-bound~\Cref{eq:generalTheorem} as follows for a general measure $\Pm_{X^n}$ 
    \begin{equation}
        \Pm_{X^n}\left(\left\lvert f-\Pm_{\bigotimes_{i=1}^n X_i}(f)\right\rvert \geq t \right) \leq 2^{\frac1\beta} \exp\left(\frac{-2t^2}{\beta \sum_{i=1}^n c_i^2}\right)  \cdot\prod_{i=2}^n \max_{x^{i-1}} H_\alpha^\frac1\alpha(\Pm_{X_i|X^{i-1}=x^{i-1}}\|\Pm_{X_i}),\label{eq:simplerTheorem2} 
    \end{equation}
    while, if $(X_1,\ldots, X_n)$ are Markovian under $\Pm_{X^n}$, \textit{i.e.}, $\Pm_{X_i|X^{i-1}} = \Pm_{X_i|X_{i-1}}$ almost surely, then the following holds:
    \begin{equation}
       \Pm_{X^n}\left(\left\lvert f-\Pm_{\bigotimes_{i=1}^n X_i}(f)\right\rvert \geq t \right) \leq 2^{\frac1\beta} \exp\left(\frac{-2t^2}{\beta \sum_{i=1}^n c_i^2}\right)\left(\prod_{i=2}^nH^\alpha_i\right)^\frac1\alpha,\label{eq:generalTheoremMarkov}
    \end{equation}
    with $\beta=\alpha/(\alpha-1)$, $H_i^\alpha = \Pm_{X_{i-1}}^\frac{1}{\beta_{i-1}}\left(H_{\alpha\alpha_i}^\frac{\beta_{i-1}}{\alpha_i}(\Pm_{X_i|X_{i-1}}\|\Pm_{X_i})\right)$, $\alpha_i>1$ for $i\geq 0$, $\beta_0=1$, $\alpha_n=1$ , and $\beta_i = \alpha_i/(\alpha_i-1)$. 
    \end{theorem}
    The proof of~\Cref{thm:main} is in~\Cref{app:proofMainThm}. 
    %
    If the function $f$ satisfies~\Cref{eq:mcDiarmidAssumption} with $c_i=\frac1n$, like in the case of the empirical mean, one obtains 
    \begin{equation}
        \Pm_{X^n}\left(\left\lvert f-\Pm_{\bigotimes_{i=1}^n X_i}(f)\right\rvert \geq t \right) \leq 2^\frac1\beta \exp\left(-n\left(\frac{2t^2}{\beta}-\frac{1}{n\alpha}\log H_\alpha(\Pm_{X^n}\|\Pm_{\bigotimes_{i=1}^n X_i})\right)\right).
    \end{equation}
    This means that if \begin{equation}
    t > \sqrt{\frac{\beta}{2n\alpha}\ln H_\alpha(\Pm_{X^n}\|\Pm_{\bigotimes_{i=1}^n X_i})}, \label{eq:thresholdEta}
    \end{equation}
    then~\Cref{thm:main} guarantees an exponential decay. If the sign of the inequality~\eqref{eq:thresholdEta} is reversed, then the bound actually becomes trivial, for $n$ large enough. The threshold behavior just described characterises the main difference of this bound with respect to existing approaches: while there are no restrictive assumptions required (other than absolute continuity of the measures at play), the bound can be trivial if the joint distribution is ``too far'' from the product of the marginals. In contrast, other approaches, like the one described in~\cite{dependentViaMartingale}, do not generally exhibit such a behaviour. Next, we will characterize the key quantity $H_\alpha(\Pm_{X^n}\|\Pm_{\bigotimes_{i=1}^n X_i})$ as a function of $n$ 
    in the concrete examples of \Cref{sec:examples}. Before doing that, a 
    few additional remarks are in order. 
    \begin{remark}[Simplification of the bound]
        The expression on the RHS of~\Cref{eq:generalTheoremMarkov} can be complicated to compute, especially due to the presence of $\{\alpha_i\}_{i=2}^\infty$. 
        Making a specific choice, which meaningfully reduces the number of parameters (\textit{i.e.}, taking $\alpha_i\to 1$ for every $i\geq 2$),~\Cref{eq:generalTheoremMarkov} boils down to the following, simpler, expression:
   \begin{equation}           
         \Pm_{X^n}\left(\left\lvert f-\Pm_{\bigotimes_{i=1}^n X_i}(f)\right\rvert \geq t \right) \leq 2^{\frac1\beta} \exp\left(\frac{-2t^2}{\beta \sum_{i=1}^n c_i^2}\right)\cdot\prod_{i=2}^n \max_{x_{i-1}} H_\alpha^\frac1\alpha(\Pm_{X_i|X_{i-1}=x_{i-1}}\|\Pm_{X_i}).\label{eq:simplerTheorem}
    \end{equation}
    Moreover,~\Cref{eq:simplerTheorem} can be re-written as follows:
    \begin{equation}
        \Pm_{X^n}\left(\left\lvert f-\Pm_{\bigotimes_{i=1}^n X_i}(f)\right\rvert \geq t \right) \leq 2^{\frac1\beta} \exp\left(\frac{1}{\beta}\left(\frac{-2t^2}{\sum_{i=1}^n c_i^2}+\sum_{i=2}^n \max_{x_{i-1}} D_\alpha(\Pm_{X_i|X_{i-1}=x_{i-1}}\|\Pm_{X_i})\right)\right). \label{eq:simplerTheoremDalpha}
    \end{equation}
    ~\Cref{eq:simplerTheoremDalpha} allows to exploit the SDPI coefficient for $D_\alpha$ (see~\Cref{rmk:SDPIDal}), which in some settings improves upon leveraging~\Cref{eq:simplerTheorem} along with hypercontractivity, see~\Cref{app:comparisonHyperContrDalpha}.
    \end{remark}
    \begin{remark}[Concentration without independence for a general event $E$]
\Cref{thm:main} can be proved in more generality. Indeed, for any measurable event $E$, one can say that, for every $\alpha>1$, \begin{equation}
    \Pm_{X^n}\left(E\right) \leq \Pm^\frac1\beta_{\bigotimes_{i=1}^n X_i}(E)\cdot H_\alpha(\Pm_{X^n}\|\Pm_{\bigotimes_{i=1}^n X_i}).\label{eq:thmGeneralE}
\end{equation} Thus, our framework is \emph{not} restricted to a McDiarmid-like setting, but it can be used to generalise \emph{any} concentration of measure approach to dependent random variables. The idea is that concentration holds when random variables are independent, namely, $\Pm_{\bigotimes_{i=1}^n X_i}(E)$ decays exponentially in $n$ under suitable assumptions. Then,~\Cref{eq:thmGeneralE} shows that a similar 
exponential decay 
holds also in the presence of dependence, as long as the measure of the joint is not ``too far'' from the product of the marginals. The ``distance'' between joint and product of the marginals is captured by the Hellinger integral $H_\alpha$. In particular, if the joint measure \emph{corresponds} to the product of the marginals, then $H_\alpha^\frac1\alpha(\Pm_{X^n}\|\Pm_{\bigotimes_{i=1}^nX_i})=1$ for every $\alpha$. Thus, taking the limit of $\alpha\to\infty$, one recovers 
\begin{equation}
\Pm_{\bigotimes X_i}(E)= \Pm_{X^n}\left(E \right) \leq \Pm_{\bigotimes X_i}(E).
\end{equation}
\end{remark}
\begin{remark}[Choice of $\alpha$]\label{rmk:tradeOff}
On the RHS of both~\Cref{eq:generalTheorem,eq:simplerTheorem}, the probability term is raised to the power $\frac{\alpha-1}{\alpha}$ and multiplied by the $\alpha$-norm of the Radon-Nikodym derivative. On the one hand, as $\alpha$ grows, the $\alpha$-norm grows as well, which increases the Hellinger integral; on the other hand, as $\alpha$ grows, $\frac{\alpha-1}{\alpha}$ tends to $1$, which reduces the probability. This introduces a trade-off between the two quantities that renders the optimisation over $\alpha$ non-trivial. We highlight that considering the limit of $\alpha\to\infty$ provides the fastest exponential decay and it recovers the probability for independent random variables. This has the cost of rendering the multiplicative constant larger, and we will discuss in detail the choice of $\alpha$ in the various examples of~\Cref{sec:examples}.\label{rmk:generalEvents}
\end{remark}

\begin{remark}[Tensorisation and McDiarmid's] Note that~\Cref{eq:generalTheorem} requires only absolute continuity as an assumption.~\Cref{eq:simplerTheorem,eq:simplerTheorem2} instead leverage tensorisation properties of $H_\alpha$. These tensorisation properties are particularly suited for Markovian settings as~\Cref{eq:simplerTheorem} shows. However, in the general case, one can still 
reduce 
$H_\alpha(\Pm_{X^n}\|\Pm_{\bigotimes_{i=1}^n X_i})$ (a divergence between $n$-dimensional measures) to $n$ one-dimensional objects, see \Cref{app:tensorisation} for details about the tensorisation of both the Hellinger integral $H_\alpha$ and R\'enyi's $\alpha$-divergence $D_\alpha$ and thus retrieve~\Cref{eq:simplerTheorem2}.
Note that~\Cref{eq:generalTheorem}
gives a natural generalisation of concentration inequalities to the case of arbitrarily dependent random variables (just like~\Cref{eq:generalTheoremMarkov}  generalises them to Markovian settings). Indeed, if $\Pm_{X^n}=\Pm_{\bigotimes_i X_i}$, then taking the limit of $\beta\to 1$ in both~\Cref{eq:generalTheorem} and~\Cref{eq:generalTheoremMarkov}, one recovers the classical concentration bound for independent random variables (see the discussion in~\Cref{rmk:generalEvents} recalling that $\beta=\alpha/(\alpha-1)$). 
\end{remark}
    \section{Applications}\label{sec:examples}
    Let us now apply~\Cref{thm:main} to four settings:
    \begin{enumerate}
        \item In~\Cref{sec:example1}, we consider a \emph{discrete-time Markovian setting}. Here, we specialise~\Cref{thm:main}  leveraging the \mbox{(hyper-)contraction} properties of the Markov kernel along with the discrete structure of the problem, 
        thus showing that 
        in certain parameter regimes our bound fares better than what the state of the art can provide;
        \item In~\Cref{sec:example2}, we consider a \emph{non-contracting Markovian setting} that  does \emph{not} admit a stationary distribution. Both these properties do not allow the application of most of the existing work in the literature. In contrast, not only our approach can be applied, but it provides exponentially decaying probability bounds, while~\cite[Theorem 1.2]{dependentViaMartingale} can only provide an upper bound that does not vanish as $n$ grows;
        \item In~\Cref{sec:example3}, we consider a \emph{non-Markovian setting} where the entire past of the process influences each step.  Here, to the best of our knowledge, we provide the first bound that exponentially decays in $n$ and has a closed-form expression, while existing approaches either cannot be employed or require the computation of complicated quantities (\textit{e.g.},~\Cref{eq:totalVarationDependentCase,eq:martonDependentCase});
        \item Finally, in~\Cref{sec:MCMC}, we apply~\Cref{thm:main} to provide error bounds on Markov Chain Monte Carlo methods. Similarly to the other settings, we propose a regime of parameters in which our results fare better and, consequently, provide an improved lower bound on the minimum burn-in period necessary to achieve a certain accuracy in MCMC.
    \end{enumerate}
  We will hereafter assume, for simplicity of exposition, that $c_i=1/n$ in~\Cref{eq:mcDiarmidAssumption} like in the case of the empirical mean. All the results hold for general $c_i$'s, but the expressions and comparisons would become more cumbersome.
 
    \subsection{Discrete-Time Markov chains} \label{sec:example1}
    
    Consider a discrete-time setting and a Markov chain $(X_n)_{n\in\mathbb{N}}$ determined by a sequence of transition matrices $(K_n)_{n\in\mathbb{N}}$. Assume that $X_1\sim P_1$ and let $X_i$ denote the random variable whose distribution is given by $P_1 K_1\ldots K_{i-1}$.\footnote{One can also see $X_i$ as the outcome of $X_{i-1}$ after being passed through the channel $K_{i-1}$.} 
    \begin{theorem}\label{thm:generalDiscrete}
        For $i\geq 1$, suppose $K_i$ is a discrete-valued Markov kernel, and let $\gamma^\star_i(\alpha)$ be the smallest parameter making it hyper-contractive, see~\Cref{sec:markovPreliminaries}.
        Then, for every function $f$ satisfying~\Cref{eq:mcDiarmidAssumption} with $c_i=\frac1n$ and every $\alpha>1$, 
        \begin{equation}
                \Pm_{X^n}(\left\lvert f-\Pm_{\bigotimes_{i=1}^n X_i}(f)\right\rvert \geq t) \leq 2^\frac1\beta \exp\left(-\frac{2nt^2}{\beta}+\sum_{i=1}^{n-1} \left( \log\left\lVert K_i^{\leftarrow}\right\rVert_{\alpha\to\gamma^\star_i(\alpha)} - \frac{1}{\bar{\gamma}_i^\star(\alpha)}\min_{j\in \text{\textit{supp}}(\Pm_{i})} \log P_i(j)\right)\right)\label{eq:generalResultHypercontractive}.
        \end{equation}
        Moreover, if the Markov kernel is time-homogeneous, \textit{i.e.}, $K_i=K$ for every $i\geq 1$, then 
        \begin{align}
            \Pm_{X^n}&(\left\lvert f-\Pm_{\bigotimes_{i=1}^n X_i}(f)\right\rvert \geq t)\notag \\
            &\leq 2^\frac1\beta \exp\left(-\frac{2nt^2}{\beta}+(n-1) \log\left\lVert K^{\leftarrow}\right\rVert_{\alpha\to\gamma^\star_K(\alpha)} - \frac{1}{\bar{\gamma}^\star_K(\alpha)}\sum_{i=1}^{n-1}\left(\min_{j\in \text{\textit{supp}}(\Pm_{i})} \log P_i(j)\right)\right)\label{eq:generalResultDiscreteHomogeneousHyper} \\
            &\leq 2^\frac1\beta \exp\left(-\frac{2nt^2}{\beta}+(n-1) \log\left\lVert K^{\leftarrow}\right\rVert_{\alpha\to {\gamma}^\star_K(\alpha)} - \frac{n-1}{\bar{\gamma}^\star_K(\alpha)}\left(\min_{i=1,\ldots,(n-1)}\min_{j\in \text{\textit{supp}}(\Pm_{i})} \log P_i(j)\right)\right)\label{eq:generalResultDiscreteHomogeneous2Hyper}.
        \end{align}
    In the above equations $\bar{\gamma}_i^\star(\alpha)$ and $\bar{\gamma}_K^\star(\alpha)$ denote the H\"older conjugates of, respectively, $\gamma_i^\star(\alpha)$ and $\gamma_K^\star(\alpha)$.
    \end{theorem}

\begin{proof}
        The main object one has to bound, according to~\Cref{thm:main}, is the following:
    \begin{equation}
    \max_{x_{i-1}} H_\alpha(\Pm_{X_i|X_{i-1}=x_{i-1}}\|\Pm_{X_i}), \text{\hspace{0.3em} with } i\geq 2.
\end{equation}
From the properties of the Markov kernel, one has that
$P_{X_i} = P_{X_{i-1}} K_{i-1}$. Furthermore, $P_{X_i|X_{i-1}=x_{i-1}}$ can be seen as $\delta_{x_{i-1}} K_{i-1}$ where $\delta_{x_{i-1}}$ is a Dirac-delta measure centered at $x_{i-1}$. Thus, recalling the definition of 
$\gamma^\star_{i-1}(\alpha)$ from 
\Cref{sec:markovPreliminaries}, 
\begin{align}
    H^\frac{1}{\alpha}_\alpha(\Pm_{X_i|X_{i-1}=x_{i-1}}\| \Pm_{X_i})  &= H^\frac{1}{\alpha}_\alpha(\delta_{x_{i-1}}K_{i-1}\| \Pm_{X_{i-1}}K_{i-1})\\ &= \left\lVert \frac{d\delta_{x_{i-1}}K_{i-1}}{d\Pm_{X_{i-1}}K_{i-1}}\right\rVert_{L^\alpha(\Pm_{X_{i-1}}K_{i-1})}  \\
    &\leq \left\lVert K_{i-1}^\leftarrow \right\rVert_{\alpha\to\gamma^\star_{i-1}(\alpha)} H_{\gamma^\star_{i-1}(\alpha)}^\frac{1}{\gamma^\star_{i-1}(\alpha)}(\delta_{x_{i-1}}\|\Pm_{X_{i-1}})\label{eq:reverseChannelRND},
\end{align}
where~\Cref{eq:reverseChannelRND} follows from~\cite[Lemma A.1]{sdpiRaginsky}. 
Moreover, for every $\kappa>1$,
 \begin{equation}
     H^\frac1\kappa_\kappa(\delta_{x_{i-1}}\|\Pm_{X_{i-1}}) = \Pm_{X_{i-1}}(\{x_{i-1}\})^\frac{1-\kappa}{\kappa} = \Pm_{X_{i-1}}(\{x_{i-1}\})^{-\frac{1}{\bar{\kappa}}},
 \end{equation}
 where $\bar{\kappa}$ denotes the H\"older's conjugate of $\kappa$ and $\Pm_{X_{i-1}}(\{x_{i-1}\})$ the measure that $\Pm_{X_{i-1}}$ assigns to the point $x_{i-1}$.
 Thus, the following sequence of steps, along with~\Cref{eq:simplerTheorem}, concludes the argument:
 \begin{align}
      \prod_{i=2}^n \max_{x_{i-1}} H^\frac1\alpha_\alpha(\Pm_{X_i|X_{i-1}=x_{i-1}}\|\Pm_{X_i}) & \leq \prod_{i=2}^{n} \max_{x_{i-1}} \left\lVert K^\leftarrow_{i-1} \right\rVert_{\alpha\to\gamma^\star_{i-1}(\alpha)} H^\frac{1}{\gamma^\star_{i-1}(\alpha)}_{\gamma^\star_{i-1}(\alpha)}(\delta_{x_{i-1}}\|\Pm_{X_{i-1}}) \\
      & = \left(\prod_{i=2}^{n} \left\lVert K^\leftarrow_{i-1} \right\rVert_{\alpha\to\gamma^\star_{i-1}(\alpha)} \right)\left(\prod_{i=2}^n  \max_{x_{i-1}}\left(\Pm_{X_{i-1}}(\{x_{i-1}\})^{-\frac{1}{\bar{\gamma}^\star_{i-1}(\alpha)}}\right)\right)\\
      &= \left(\prod_{i=1}^{n-1} \left\lVert K^\leftarrow_i \right\rVert_{\alpha\to\gamma^\star_{i}(\alpha)} \right)\left(\prod_{i=1}^{n-1}  \left(\min_{x_i}\Pm_{X_i}(\{x_i\})\right)^{-\frac{1}{\bar{\gamma}^\star_{i}(\alpha)}}\right).
 \end{align}
 Moreover, given the discrete setting, one can replace the measure $\Pm_{X_i}$ with the corresponding pmf which is denoted by $P_i$.
    \end{proof}
        If 
        the Markov kernel $K_i$ is only contractive (and not hyper-contractive), then 
        $\gamma^\star_i(\alpha)=\alpha$ and $\bar\gamma^\star_i(\alpha)=\beta$, which allows to simplify \Cref{eq:generalResultHypercontractive,eq:generalResultDiscreteHomogeneousHyper,eq:generalResultDiscreteHomogeneous2Hyper}. 
   In this case, 
if 
   \begin{align}
       t^2 &\geq \frac{\beta}{2}\frac{n-1}{n}\log\left\lVert K^\leftarrow \right\rVert_{\alpha\to\alpha} - \frac{n-1}{2n} \left(\min_i \min_j \log P_i(j)\right)\\
       &= (1+o_n(1))\left(\frac{\beta}{2}\log\left\lVert K^\leftarrow \right\rVert_{\alpha\to\alpha}^\beta -\frac12  \left(\min_i \min_j \log P_i(j)\right) \right), 
   \end{align}
   then \Cref{thm:generalDiscrete} gives 
   exponential (in $n$) concentration even in the case of dependence. 
   
    \begin{remark}[Hypercontractivity, SDPI and R\'enyi's divergences]\label{rmk:SDPIDal}
    Another perspective naturally 
    stems from~\Cref{eq:simplerTheoremDalpha}. Indeed, similarly to~\Cref{thm:generalDiscrete}, one has 
        \begin{align}
        \Pm_{X^n}(\left\lvert f-\Pm_{\bigotimes_{i=1}^n X_i}(f)\right\rvert \geq t)  &\leq 2^\frac1\beta \exp\left(-\frac1\beta\left(2nt^2-\sum_{i=2}^n \max_{x_{i-1}} D_\alpha(\Pm_{X_i|X_{i-1}=x_{i-1}}\|\Pm_{X_i})\right)\right)\\
        &\leq 2^\frac1\beta \exp\left(-\frac1\beta\left(2nt^2-\eta_\alpha(K)  \sum_{i=2}^n\max_{x_{i-1}} D_\alpha(\delta_{x_{i-1}}\|\Pm_{X_{i-1}})\right)\right) \label{eq:SDPIDalpha} \\
        & = 2^\frac1\beta \exp\left(-\frac1\beta\left(2nt^2+\eta_\alpha(K)  \sum_{i=2}^n\min_{x_{i-1}} \log P_{i-1}(x_{i-1})\right)\right)\label{eq:SDPIDalphaBound}.
        \end{align}
    We remark that 
    $D_\alpha$ can also be hyper-contractive with respect to some Markovian operators, meaning that in~\Cref{eq:SDPIDalpha} for instance, one could consider $D_\gamma(\delta_{x-1}\|\Pm_{X_{i-1}})$ with $\gamma < \alpha$ (this is equivalent to hyper-contractivity of Markov operators,~\cite[Section IV]{logSobolevContraction}). One such example is the Ornstein–Uhlenbeck channel with noise parameter $t$, cf.~\cite[Eq.  3.2.37]{concentrationMeasureII}, for which one can prove hyper-contractivity with respect to $D_\alpha$~\cite[Theorem 3.2.3]{concentrationMeasureII}. Moreover, in some settings, leveraging SDPIs for $D_\alpha$ can provide an improvement over~\Cref{thm:generalDiscrete}, see \Cref{app:comparisonHyperContrDalpha} for a detailed
    comparison.
    \end{remark}
   
    We now compare the concentration bound provided by \Cref{thm:generalDiscrete} with existing bounds in the literature. In this section, the comparison concerns a general Markov kernel, and the explicit calculations for a binary kernel are deferred to \Cref{app:ex1}. 

    \subsubsection{Comparison with \cite[Theorem 1.2]{dependentViaMartingale}}  Let us consider the same setting as in~\Cref{thm:generalDiscrete}. Then, \cite{dependentViaMartingale} gives %
    \begin{equation}
        \mathbb{P}(|f-\Pm_{X^n}(f)|\geq t ) \leq 2 \exp\left(-\frac{nt^2}{2 M_n^2}\right), \label{eq:kontorovichGeneral}
    \end{equation}
    where $M_n = \max_{1\leq i\leq n-1} \left(1+\sum_{j=i}^{n-1} \prod_{k=i}^j \eta_{KL}(K_k) \right)$ and we recall that $\eta_{KL}(K_i)=\sup_{x,\hat{x}}TV(K_i(\cdot|x),K_i(\cdot|\hat{x}))$ is the contraction coefficient of the Markov kernel $K_i$. First, note that, 
    if the random variables are independent and thus $\Pm_{X^n}=\Pm_{\bigotimes X_i}$, then~\Cref{eq:kontorovichGeneral} reduces to $$\Pm_{X^n}(E) \leq 2 \exp(-nt^2/2),$$ while~\Cref{thm:generalDiscrete} with $\gamma^\star_K(\alpha)=\alpha\to\infty$ and $\bar\gamma^\star_K(\alpha)=\beta\to 1$ recovers McDiarmid's inequality with the correct constant in front of $n$, \textit{i.e.},
    $$\Pm_{X^n}(E) \leq 2 \exp(-2nt^2).$$
    
    Assume now that the Markov kernel is time-homogeneous and has a contraction coefficient $\eta_{TV}(K) < 1$. Then, $M_n = \max_{1\leq i\leq n-1} \frac{1-\eta_{TV}(K)^{n+1-i}}{1-\eta_{TV}(K)}= \frac{1-\eta_{TV}(K)^n}{1-\eta_{TV}(K)}$. For compactness, define $P_{i^\star}(j^\star):=\min_i\min_j P_i(j)$. Making a direct comparison, one has that if 
    \begin{align}
        t^2 &> \frac{n-1}{n} \left(\frac{2M_n^2}{4M_n^2-\beta}\right)\log\frac{\left\lVert K^\leftarrow \right\rVert_{\alpha\to\alpha}^\beta}{P_{i^\star}(j^\star)} \\
        &= (1+o_n(1)) \left(\frac{2}{4-\beta(1-\eta_{TV}(K))^{2}}\right)\log\frac{\left\lVert K ^\leftarrow \right\rVert_{\alpha\to\alpha}^\beta}{P_{i^\star}(j^\star)}, \label{eq:comparisonKontorovichTimeHom}
    \end{align}
    then~\Cref{eq:generalResultDiscreteHomogeneous2Hyper} (with $\gamma^\star_K(\alpha)=\alpha$ and $\bar\gamma^\star_K(\alpha)=\beta$) provides a faster exponential decay than~\Cref{eq:kontorovichGeneral}. The explicit calculations for the special case of a binary kernel are provided in~\Cref{app:ex1Comp1}.

 One can then consider the limit of $\alpha\to\infty$ which renders $\beta\to 1$. On the one hand, this implies a larger multiplicative coefficient (as $H_\alpha$ grows with $\alpha$, see~\Cref{rmk:tradeOff}) and, consequently, it increases the minimum value of $t$ one can consider in~\Cref{eq:comparisonKontorovichTimeHom}. On the other hand, it guarantees a faster exponential decay in~\Crefrange{eq:generalResultDiscreteHomogeneousHyper}{eq:generalResultDiscreteHomogeneous2Hyper}. 
 In fact, as $\beta\to 1$, the RHS of~\eqref{eq:generalResultDiscreteHomogeneous2Hyper} scales as $\exp(-2nt^2(1+o_n(1)))$ for large enough $t$, which matches the behavior of~\Cref{eq:mcDiarmidStatement}. 
 Let us highlight that, to the best of our knowledge, our approach is the \emph{first} to recover the \emph{same exponential decay rate} obtained in the independent case, even in the presence of correlation among the random variables. 
 
We remark that the convergence results in \Cref{thm:generalDiscrete} and \Cref{eq:kontorovichGeneral} are with respect to different constants: \Cref{thm:generalDiscrete} considers the concentration of $f$ around $\Pm_{\bigotimes_{i=1}^n X_i}(f)$, while~\Cref{eq:kontorovichGeneral} around $\Pm_{X^n}(f)$. 
However, given the faster rate of convergence guaranteed by our framework -- for this example and, even more impressively so, for the one in~\Cref{sec:example2} -- the mean of $f$ under the product of the marginals might be regarded as a natural object to consider when proving concentration results for these processes.  
    This hypothesis is corroborated by the approach presented in~\cite{dependentViaStatAndChange}, where concentration for stationary Markov chains is provided around the mean with respect to the stationary measure $\pi$. Indeed, when $X_1\sim \pi$, the product of the marginals reduces to the tensor product of the stationary measure $\pi^{\otimes n}$. 
 %
    To make a direct comparison with~\cite{dependentViaMartingale}, one can leverage~\cite[Proposition 1.8]{ledoux2001concentration} (reproduced in \Cref{app:cmm}) and either reduce both results to concentration bounds around the median, or transform~\Cref{thm:main} in a result on concentration around the mean with respect to $\Pm_{X^n}$. This would introduce additional terms, rendering the comparison cumbersome and outside the scope of this work. We however perform said comparison explicitly in a different setting (see~\Cref{sec:example2}).

     \subsubsection{Comparison with \cite[Theorem 1]{dependentViaStatAndChange}} 
     
\cite{dependentViaStatAndChange} considers a more restricted setting in which $f$ is the sum of $n$ bounded functions that separately act on each of the $n$ random variables, \textit{i.e.}, $f=\sum_{i=1}^n f_i(X_i)$ with $f_i \in [a_i,b_i]$.\footnote{This choice of $f$ in~\Cref{thm:main} transforms the result in a generalisation of Hoeffding's inequality to dependent settings. It is easy to see that if $f_i\in [a_i,b_i]$ then~\eqref{eq:mcDiarmidAssumption} holds with $c_i = (b_i-a_i)$ and the statement follows.} By assuming further that $a_i=0$ and $b_i=\frac1n$ for every $i$ (\textit{e.g.}, empirical mean), we are in a setting in which~\Cref{thm:generalDiscrete} holds as it is. Moreover, the setting in \cite{dependentViaStatAndChange} requires the Markov chain to be time-homogeneous and admit a stationary distribution $\pi$ (notice that none of these assumptions are necessary for Theorems~\ref{thm:main} and~\ref{thm:contractingMarkov} to hold). In this case~\cite[Theorem 1]{dependentViaStatAndChange} gives:
\begin{equation}
    \Pm_{X^n}\left(\left|f -  \pi^{\otimes n}(f)\right| \geq t \right) \leq 2 \exp\left(-\frac{1-\lambda}{1+\lambda}2nt^2\right),\label{eq:generalDiscreteStationary}
\end{equation}
where $(1-\lambda)$ denotes the absolute spectral gap of the Markov chain, see~\cite[Definition 19]{dependentViaStatAndChange}, which characterizes the speed of convergence to the stationary distribution. Comparing with~\Cref{eq:generalResultDiscreteHomogeneous2Hyper} with $\gamma^\star_K(\alpha)=\alpha\to\infty$ and $\bar\gamma^\star_K(\alpha)=\beta\to 1$,\footnote{As mentioned earlier, this optimizes the rate of convergence, at the expense of the value of $t$ from which we obtain an improvement.} one has that if
\begin{equation}
    t^2 > (1+o_n(1))\frac{1+\lambda}{4\lambda}\log\frac{\left\lVert K^\leftarrow \right\rVert_{\infty\to\infty}}{P_{i^\star}(j^\star)},
\end{equation}
then~\Cref{thm:generalDiscrete} provides a faster decay than~\Cref{eq:generalDiscreteStationary}. A more explicit comparison in which the absolute spectral gap is computed for a binary Markov kernel can be found in~\Cref{app:ex1Comp2}.
   
\subsubsection{Comparison with \cite{Marton1996BoundingB}}    
    Let us now derive the corresponding result of concentration around the median in order to compare with~\cite{Marton1996BoundingB}. Leveraging~\cite[Proposition 1.8]{ledoux2001concentration} (see also \Cref{app:cmm}) along with~\Cref{eq:generalResultDiscreteHomogeneous2Hyper} (again, with $\gamma^\star_K(\alpha)=\alpha\to\infty$ and $\bar\gamma^\star_K(\alpha)=\beta\to 1$), we have 
    \begin{equation}       
    \Pm_{X^n}\left(|f-m_f|\geq t\right) 
 \leq 2 \exp \left(-2n\left(t-\sqrt{\frac{\ln4+C_n}{2n}}\right)^2+C_n\right),\label{eq:convergenceMedianGeneral}
    \end{equation}
    where $C_n=(n-1)\log\left(\left\lVert K^\leftarrow\right\rVert_{\infty\to\infty}/P_{i^\star}(j^\star)\right)$. 
    This also implies (see, \textit{e.g.},~\cite[Proposition 1.3]{ledoux2001concentration} or~\cite[Theorem 3.4.1]{concentrationMeasureII}) that, if $t>\sqrt{\frac{\log4+C_n}{2n}}$ and given any event $E$ such that $\Pm_{X^n}(E)\geq \frac12$, then 
    \begin{equation}
        \begin{split}            
    \Pm_{X^n}\left(E^c_t\right) &\leq 2 \exp \left(-2n\left(t-\sqrt{\frac{\log4+C_n}{2n}}\right)^2+C_n\right)\\
    &= \frac12\exp\left(-2nt^2+2t\sqrt{2n\left(\log4+C_n\right)}\right)\label{eq:ourGeometricSimpleChannel},
        \end{split}
    \end{equation}
    where $E_t = \{y\in\X^n : d(x,y)\leq t \text{ for some } x\in E\}$ and $d$ denotes the normalised Hamming metric. 
    
    Let us denote 
    $a:=1-\max_i \max_{x,\hat{x}} TV(\Pm_{X_i|X_{i-1}=x},\Pm_{X_i|X_{i-1}=\hat{x}}).$
    Hence, assuming $t>\frac1a\sqrt{\log(1/\Pm_{X^n}(E))/n}$, \cite[Proposition 1 \& Proposition 4]{Marton1996BoundingB} give
    \begin{align}           
    \Pm_{X^n}(E^c_t) &\leq \exp\left(-2n\left(at -\sqrt{\frac{\log(1/\Pm_{X^n}(E))}{2n}}\right)^2\right) \\
    & \leq \exp\left(-2nt^2a^2+2ta\sqrt{2n\log2}\right)\label{eq:martonGeometricDiscreteGeneral}.
    \end{align}
    Ignoring multiplicative constants and comparing the exponents of~\eqref{eq:ourGeometricSimpleChannel} and~\eqref{eq:martonGeometricDiscreteGeneral}, one can see that, whenever
     \begin{equation}
    t \geq \left(1+o_{n}(1)\right)\frac{\sqrt{2\log\left(\frac{\left\lVert K^\leftarrow\right\rVert_{\infty\to\infty}}{P_{i^\star}(j^\star)}\right)}}{(1-a^2)}
,
    \end{equation} 
    then our approach improves upon \cite{Marton1996BoundingB}. 
    
    The Dobrushin coefficient of the kernel, captured by the quantity $(1-a)$, measures the degree of dependence of the stochastic model. The smaller $1-a$ is, the less ``dependent'' the model is. If $\Pm_{X^n}$ reduces to a product distribution then $a=1$. In this case, \Cref{eq:ourGeometricSimpleChannel} 
    boils down to McDiarmid's inequality. In contrast, the larger $1-a$ is, the worse the behavior of~\Cref{eq:ourGeometricSimpleChannel}. If $a=0$, then the Markov chain is not contracting and violates the assumption of~\cite{Marton1996,Marton1996BoundingB,Marton1998,Marton2003}. Our approach, instead, can still provide meaningful results, as we will see in~\Cref{sec:example2}. A more explicit comparison for the case of a binary kernel can be found in~\Cref{app:ex1Comp3}.
    
    \subsection{A non-contracting Markov chain}\label{sec:example2}
     In order to provide concentration for Markov chains, existing work 
    requires either contractivity of the Markov chain
    ~\cite{Marton1996}, stationarity
    ~\cite{dependentViaStatAndChange, Marton1996BoundingB} or some form of mixing~\cite{samsonConcentration, Marton1996BoundingB}.
    A well-known Markov chain that evades most of these concepts is the SSRW. Suppose to have a sequence of i.i.d.\ Rademacher random variables, \textit{i.e.}, $\mathbb P(X_i =  -1)=\mathbb P(X_i= +1)=1/2$, for $i\geq 1$; the initial condition is $X_0=0$ w.p. $1$. Then, a SSRW is the Markov chain $(S_i)_{i\in\mathbb N}$ defined as $S_i = S_{i-1} + X_i$. 
    This Markov chain does not admit a stationary distribution, it is not contracting, but it is expanding (in this case, both towards the positive and the negative axes of the real line). Let us denote with $K_i$ the kernel at step $i$. Then, at each step $i>2$, one can always find two different realisations of $S_{i-1}$, let us call them $s_1,s_2$, such that  $\text{supp}(K_i(\cdot|s_1))\cap \text{supp}(K_i(\cdot|s_2))=\emptyset$, \textit{i.e.}, the support of $S_i$ is constantly growing. 
    This implies that the approaches in~\cite{Marton1996,dependentViaStatAndChange} cannot be employed, while~\cite[Theorem 1]{dependentViaMartingale} yields:  
    \begin{equation}
    \mathbb{P}(\left\lvert f-\Pm_{X^n}(f)\right\rvert \geq t) \leq 2\exp\left(-\frac{t^2}{2n}\right). \label{eq:ssrwKontorovich}
    \end{equation}
    A meaningful regime (given also the expanding nature of the Markov chain along the integers)\footnote{The standard deviation of $S_n$ is $\sqrt{n}$, hence it is expected that $S_n$ is $\pm O(\sqrt{n})$.} 
    arises when considering $t \gtrsim \sqrt{n}$.
    Let us now compute the Hellinger integral in this specific setting. This is done in the lemma below, proved in~\Cref{app:proofLemmaSSRW}.
    \begin{lemma}\label{lemma:hAlphaSSRW}
    Let $i\geq 1$, $x\in \text{supp}(S_{i-1})$, $0\leq j\leq i$, and $\alpha\geq 1$. Then,
    \begin{equation}   2^{\frac1\beta\left(-1+ i(1- h_2(\frac{i+1}{2i})+ \frac{1}{2} \log_2\left(\frac{\pi}{2}\left(\frac{i^2-1}{i}\right)\right)\right)} \leq H_\alpha^{\frac1\alpha}(\Pm_{S_i|S_{i-1}=x}\|\Pm_{S_n}) \leq 2^{i \frac{1}{\beta}-1+\frac1\alpha},\end{equation} where $h_2(x)=-x\log_2(x)-(1-x)\log_2(1-x)$ denotes the binary entropy. Thus one has that:
    \begin{equation} \frac{n-2}{4\beta} \leq \log_2 H_\alpha^\frac1\alpha(\Pm_{S^n}\|\Pm_{\bigotimes_{j=1}^n S_j}) \leq \frac{n(n-1)}{2\beta}.
    \end{equation}
    \end{lemma}

Combining~\Cref{lemma:hAlphaSSRW} and~\Cref{thm:main}, one has that
    \begin{equation}
\mathbb{P}\left(\left\lvert f-\Pm_{\bigotimes_{i=1}^n S_i}(f)\right\rvert \geq t \right)\leq 2^\frac1\beta \exp\left(\frac{-2nt^2}{\beta}+ \frac{n(n-1)}{2\beta}\ln 2\right).\label{eq:ssrwOurBound}       
    \end{equation}
    It is easy to see that, whenever $t > \frac{\sqrt{(n-1)\ln(2)}}{2}$, \Cref{eq:ssrwOurBound} gives an exponential decay. For instance, choosing $t =\sqrt{n}$,  
    one retrieves
    \begin{equation}
    \mathbb{P}\left(\left\lvert f-\Pm_{\bigotimes_{i=1}^n S_i}(f)\right\rvert \geq \sqrt{n}\right) \leq 2^\frac1\beta  \exp\left(-\frac{n^2}{\beta}\left(2 - \frac{\ln2}{2} +\frac{\ln2}{2n}\right)\right).
    \end{equation}
    In contrast, the same choice in~\Cref{eq:ssrwKontorovich} gives
    \begin{equation}
    \mathbb{P}(\left\lvert f-\Pm_{S^n}(f)\right\rvert \geq \sqrt{n}) \leq 2\exp(-1/2).
    \end{equation}
    More generally, selecting $t$ of order $\sqrt{n}$ suffices to achieve an exponential decay in~\Cref{eq:ssrwOurBound}, while to obtain a similar speed of decay in~\Cref{eq:ssrwKontorovich} $t$ needs to be at least of order $n$. The approach advanced in this work can, thus, not only be employed in settings where most of the other approaches fail (\textit{e.g.},~\cite{Marton1996,Marton1996BoundingB,dependentViaStatAndChange}), but it also brings a significant improvement over the rate of decay that one can provide. It is worth noticing that~\Cref{eq:ssrwOurBound} provides concentration around $\Pm_{\otimes_{i=1}^n S_i}(f)$ while~\Cref{eq:ssrwKontorovich} provides concentration around $\Pm_{S^n}(f)$. In order to provide a more fair comparison between the two approaches, we will now provide a bound on concentration around $\Pm_{S^n}(f)$ via~\Cref{thm:main}.
    In particular, one can prove the following (see~\Cref{app:boundOnMeans} for the proof).
    \begin{lemma}\label{lemma:differenceOfMeans}
        Let $\alpha>1$. Denote with $t_\alpha = \sqrt{\frac{\beta\ln(H_\alpha(\Pm_{X^n}\|\Pm_{\otimes X_i}))} {2\alpha n}}$. Then, the following holds true:
        \begin{equation}
            \left| \Pm_{X^n}(f)- \Pm_{\otimes_{i=1}^n X_i}(f)\right| \leq t_\alpha + \frac{\sqrt{\beta}}{2^\frac1\alpha\sqrt{\frac{2n}{\alpha} \ln(H_\alpha(\Pm_{X^n}\|\Pm_{\otimes_{i=1}^n X_i}))}}.
        \end{equation}
    \end{lemma}
    Moreover, one can leverage~\Cref{lemma:hAlphaSSRW} to prove that $\frac{\sqrt{\beta}}{2^\frac1\alpha\sqrt{\frac{2n}{\alpha} \ln(H_\alpha(\Pm_{X^n}\|\Pm_{\otimes_{i=1}^n X_i}))}} \leq \frac{2^\frac{1}{\beta}\beta}{n \sqrt{\ln(2)(2-\frac4n)}} = o_n(1)$.
    Consequently, one has that
    \begin{align}
    \left|f - \Pm_{S^n}(f) \right| &=  \left|f 
 - \Pm_{S^n}(f) - \Pm_{\otimes_{i=1}^n S_i}(f)+ \Pm_{\otimes_{i=1}^n S_i}(f)\right| \\
 &\leq \left|f- \Pm_{\otimes_{i=1}^n S_i}(f) \right| + \left|\Pm_{X^n}(f)- \Pm_{\otimes_{i=1}^n S_i}(f)\right| \\
 &\leq \left|f- \Pm_{\otimes_{i=1}^n S_i}(f) \right| + t_\alpha+o_n(1). \label{eq:meanComparisons}
\end{align}
~\Cref{eq:meanComparisons} together with~\Cref{lemma:hAlphaSSRW} and~\Cref{thm:main} imply that:
\begin{align}
    \mathbb{P}\left(\left|f- \Pm_{S^n}(f) \right|\geq  t+t_\alpha+o_n(1) \right) &\leq  \mathbb{P}\left(\left|f- \Pm_{\otimes_{i=1}^n S_i}(f) \right|\geq  t \right) \\
    &\leq 2^\frac1\beta H_\alpha^\frac1\alpha(\Pm_{S^n}\|\Pm_{\otimes_{i=1}^n S_i}) \exp\left(-\frac{2n t^2}{\beta}\right).
\end{align}
Hence, selecting $\tilde{t}=2\sqrt{n}$,~\Cref{thm:main} and~\Cref{lemma:hAlphaSSRW} lead to 
\begin{equation}
     \mathbb{P}\left(\left|f- \Pm_{S^n}(f) \right|\geq \tilde{t} \right) \leq 2^\frac1\beta  \exp\left(-\frac{n^2}{\beta}\left(2 - \frac{\ln2}{2} +\frac{\ln2}{2n}\right)\right), 
\end{equation}
while, similarly to before,~\cite[Theorem 1]{dependentViaMartingale} (see~\Cref{eq:ssrwKontorovich}) leads to:
\begin{equation}
     \mathbb{P}\left(\left|f- \Pm_{S^n}(f) \right|\geq \tilde{t} \right) \leq 2\exp(-2).
\end{equation}
    \subsection{A non-Markovian Process} \label{sec:example3}
    Next, we consider a non-Markovian setting in which each step of the stochastic process depends on its entire past: 
    \begin{equation}
    X_n = \begin{cases}
    +1, &\text{ with probability }  \sum_{i=0}^{n-1}p_iX_i, \\
    -1, &\text{ with probability } 
    1 - \sum_{i=0}^{n-1}p_iX_i,
    \end{cases}
    \end{equation}
    for $n\ge 1$, $p_i>0$ and $X_0=+1$ with probability 1.
    The choice of the parameters $p_i$ is arbitrary (given the constraint that $\sum_{i=0}^{n-1} p_i < 1$) but for concreteness we will set $p_i = 2^{-i-1}$ for every $i\geq 0$. Then, $\mathbb{P}_{X_1}(1|x_0)=\frac12=\mathbb{P}_{X_1}(-1|x_0)$ and for each $n\geq 1$,
    $$\mathbb{P}_{X_n}(1|x_0^{n-1})=\frac12+ \sum_{i=1}^{n-1}p_ix_i= \sum_{i=0}^{n-1}x_i 2^{-i-1}= 1- \mathbb{P}_{X_n}(-1|x_0^{n-1}),$$ 
    with $x_0=1$.
    Consequently, following the calculations detailed in~\Cref{app:computationsNonMarkovianEx}, we have
    \begin{align}\label{eq:det1}
    H_\alpha(\Pm_{X_n}(\cdot|x_0^{n-1})\| (1/2,1/2)) 
    &< 2\sum_{j=0}^{\lfloor \frac{\alpha}{2}\rfloor} \binom{\lfloor \alpha \rfloor}{2j} \left(2\sum_{i=1}^{n-1}p_ix_i\right)^{2j},
    \end{align}
which, as $p_i = 2^{-i-1}$, gives 
    \begin{equation}\label{eq:det2} \max_{x_1^{n-1}}H_\alpha(\Pm_{X_n}(\cdot|x_0^{n-1})\| (1/2,1/2))  < 2^{\alpha}.
    \end{equation} Thus, $H_\alpha^\frac1\alpha(\Pm_{X^n}\|(1/2,1/2)^{\otimes n})< 2^{n-1}$ 
    and an application of \Cref{thm:main}, as stated in~\Cref{eq:simplerTheorem}, yields:
    \begin{equation}   
    \mathbb{P}\left(\left\{\left\lvert f-\Pm_{\bigotimes_{i=1}^n X_i}(f)\right\rvert \geq t\right\} \right) \leq \inf_{\beta>1} 2^\frac1\beta \exp\left(-\frac{2n}{\beta}\left(t^2-\frac{n-1}{n}\frac{\beta\ln2}{2}\right)\right),\label{eq:boundArbitraryDepend}
    \end{equation}
    with exponential decay whenever
    \begin{equation}
    t^2 > (1+o_n(1))\frac{\beta\ln2}{2}.
    \end{equation}
    
    Like before, choosing a larger $\beta$ slows down the exponential decay, but it reduces the multiplicative coefficient introduced via $H_\alpha$. For $n$ 
    and $t$ large enough, one can pick $\beta\to 1$ and retrieve a McDiarmid-like exponential decay.
     Given that this setting does not characterise a Markovian dependence (at each step the stochastic process depends on its entire past), one cannot employ the technique described in~\cite{dependentViaStatAndChange} or in~\cite{Marton1996}. One can, however, employ the technique described in~\cite{samsonConcentration} and~\cite{dependentViaMartingale}. Both these approaches require the computation of $\{\bar{\theta}_{ij}\}_{1\leq i < j \leq n}$ with
    \begin{equation}
      \bar{\theta}_{ij}:=   \sup_{x^{i-1}, w,\hat{w}} \left\lVert \mathcal{L}(X_j^n|X^i=(x^{i-1}, w)) - \mathcal{L}(X_j^n|X^i=(x^{i-1},\hat{w})) \right\rVert_{TV}, \label{eq:totalVarationDependentCase}
    \end{equation}
    where $\mathcal{L}(X_j^n|X^i=(x^{i-1},w))$ denotes the conditional distribution of $X_j^n$ given $X^i=(x^{i-1},w)$. The $\bar{\theta}$'s are then organised in $n\times n$ upper-triangular matrices whose norms are computed in order to provide an upper bound on the probability of interest. In particular,~\cite{samsonConcentration} requires the $\ell_2$-norm, while~\cite{dependentViaMartingale} requires the operator norm of the matrix.  
    Similarly, to employ the technique provided in~\cite{Marton2003} one would need to compute the following quantity
    \begin{equation}
    \bar{C} = \max_{1\leq i \leq n}\sup_{\substack{x^{i-1}\in\X^{i-1}\\ w,\hat{w}\in\X}}\inf_{\pi \in \Pi(\mathcal{L}(X^n|x^i,w), \mathcal{L}(X^n|x^i,\hat{w}))} \pi(d). \label{eq:martonDependentCase}
    \end{equation}
    Here, $d$ denotes the normalised Hamming metric and $\Pi(\mathcal{L}(X^n|x^i,w), \mathcal{L}(X^n|x^i,\hat{w}))$ represents the set of all the couplings defined on $\X^n\times \X^n$ such that the corresponding marginals are $\mathcal{L}(X^n|x^i,w)$ and $\mathcal{L}(X^n|x^i,\hat{w})$. However, computing any of these objects in practice can be complicated even in simple settings. In contrast, with the framework proposed here and thanks to the tensorisation properties of the Hellinger integral, one can easily bound the information measure and provide an exponentially  decaying probability (whenever the probability of the same event under independence decays exponentially and for opportune choices of the parameters).
    
    \subsection{Markov Chain Monte Carlo}\label{sec:MCMC}
    An intriguing application of the method proposed in~\cite{dependentViaStatAndChange} consists in providing error bounds for Markov Chain Monte Carlo methods. For instance, assume that one is trying to estimate the mean $\pi(f)$  for some function $f$ and some measure $\pi$ which cannot be directly sampled. A common approach consists in considering a Markov chain $\{X_i\}_{i\geq 1}$ whose stationary distribution is $\pi$ and estimating $\pi(f)$ via empirical averages of samples $\{X_i\}_{n_0+1}^{n_0+n}$, where $n_0$ characterises the so-called ``burn-in period''. This period ensures that enough time has passed and the Markov chain is sufficiently close to the stationary distribution $\pi$ before sampling from it.~\cite[Theorem 12]{dependentViaStatAndChange} gives that 
    \begin{equation}
    \mathbb{P}\left(\frac1n \sum_{i=1}^n f(X_{n_0+i})-\pi(f) > t \right) \leq C(\nu,n_0,\alpha)\exp\left(-\frac1\beta \cdot \frac{1-\max\{\lambda_r,0\}}{1+\max\{\lambda_r,0\}}\cdot \frac{2nt^2}{(b-a)^2} \right), \label{eq:MCMCFan}
    \end{equation}
    where $f:\X\to [a,b]$ is uniformly bounded, $\lambda_r$ represents the right-spectral gap (see~\cite[Definition 20]{dependentViaStatAndChange}), $\alpha\in(1,+\infty)$, $\beta$ denotes its H\"older's conjugate and $C$ is a constant depending on the burn-in period $n_0$, the Radon-Nikodym derivative between the starting measure $\nu$ and the stationary measure $\pi$, and $\alpha$.
    Using the tools provided in this work, we obtain:
      \begin{align}
    \mathbb{P}\left(\frac1n \sum_{i=1}^n f(X_{n_0+i})-\pi(f) > t \right) &\leq \exp\left(-\frac{2nt^2}{\beta(b-a)^2}\right) H_\alpha^\frac1\alpha(\nu K^{n_0}\|\pi) \prod_{i=2}^n \max_{x_{n_0+i-1}} H_\alpha^\frac1\alpha(K(\cdot|x_{n_0+i-1})\|\pi)\notag\\
    &\leq C(\nu,n_0,\alpha)\exp\left(-\frac{2nt^2}{\beta(b-a)^2}\right)\max_{x} \pi(\{x\})^{-\frac{n-1}{\beta}} , \label{eq:MCMCOurn0}
    \end{align}
   where~\Cref{eq:MCMCOurn0} follows from the fact that $H^\frac1\alpha_\alpha(\nu K^{n_0}\|\pi)$ represents the $L^\alpha(\pi)$-norm of the Radon-Nikodym derivative and can thus be bounded like in~\cite[Theorem 12]{dependentViaStatAndChange}.
    The idea behind the result is as follows. Given that one is trying to estimate the mean of $f$ under $\pi$ using empirical averages, if one had samples taken in an i.i.d.\ fashion from $\pi$, the exponential convergence would be guaranteed. However, the issue is that one does \emph{not} have access to samples of $\pi$. Thus, changing the measure to an $n$-fold tensor product of $\pi$, one can still guarantee an exponential decay at the cost of a multiplicative price depending on how far the samples are from the stationary distribution.
    By making a direct comparison,  assuming $\lambda_r>0$, one can see that if \begin{align}
        t^2 &\geq \frac{n-1}{n}\frac{(b-a)^2}{2}\frac{1+\lambda_r}{2\lambda_r} \log \left(\frac{1}{\min_x \pi(\{x\})}\right)\label{eq:thresholdEtaMCMC}\\
        &=\left(1+o_n(1)\right)\frac{(b-a)^2}{2}\frac{1+\lambda_r}{2\lambda_r} \log \left(\frac{1}{\min_x \pi(\{x\})}\right),\notag
    \end{align} then the RHS of~\Cref{eq:MCMCOurn0} decays faster than the RHS of~\Cref{eq:MCMCFan}. A comparison for the binary symmetric channel, with the computations of all the parameters,  can be found in~\Cref{app:MCMCBinary}. 

    Similarly to~\cite{dependentViaStatAndChange}, one can also show that an exponential decay is guaranteed in~\Cref{eq:MCMCOurn0} if $n_0 = \Omega(\log n)$. Furthermore, as~\Cref{eq:MCMCOurn0} improves the exponential decay for $t$ satisfying~\Cref{eq:thresholdEtaMCMC}, in the same regime the induced lower bound over $n_0$ will be improved as well. Finally, we highlight that  
    \Cref{thm:main} applies to a much larger family of functions $f$ than what can be handled by  \cite{dependentViaStatAndChange}.

   \section{Conclusions}\label{sec:conclusions}
   We introduced a novel approach to the concentration of measure for dependent random variables. The generality of our framework allows to consider arbitrary kernels without requiring either stationarity (as opposed to~\cite{dependentViaStatAndChange}) or contractivity (as opposed to~\cite{Marton1996,Marton1996BoundingB}). Moreover, our technique applies to any family of functions which is known to concentrate when the random variables are actually independent. Said technique is employed and compared to the state of the art in four different settings: finite-state Markov chains, a non-contractive one (the SSRW), a non-Markovian process, and Monte Carlo Markov Chain methods. In each of these settings, we provide a regime of parameters in which we guarantee a McDiarmid-like decay and improve over existing results. 
    The improvement is the most striking in the case of the SSRW, where the only (closed-form) alternative approach gives a \emph{constant} probability of deviation from the average, as opposed to the \emph{exponentially decaying} probability guaranteed by our framework.
    The bounds provided display a threshold phenomenon depending on the accuracy $t$, \textit{i.e.}, one can show concentration only for values of $t$ larger than a threshold depending on the Hellinger integral (and its scaling with respect to the number of variables $n$). 
    Consider for instance~\Cref{eq:boundArbitraryDepend}: if $t^2> \frac{\beta}{2n} \ln H^\frac{1}{\alpha}_\alpha \approx (1+o_n(1))(\beta\ln(2)/2)$, then the exponent is negative and one has exponential concentration, otherwise the exponent becomes positive and the bound trivialises to something larger than $1$. We believe this to be an artifact of the analysis and not an intrinsic property of the concentration of measure phenomenon.
    \section*{Acknowledgments}
The authors are partially supported by the 2019 Lopez-Loreta Prize. They would also like to thank Professor Jan Maas for providing valuable suggestions and comments on an early version of the work.
\bibliographystyle{IEEEtran}
\bibliography{sample}


\appendices

\crefalias{section}{appendix}
\crefalias{subsection}{appendix}

\section{Proof of~\Cref{thm:main}}\label{app:proofMainThm}

\begin{proof}
Assume that $E=\{|f-\Pm_{\bigotimes_{i=1}^n X_i}(f)|\geq t\}$. Then, one has that
\begin{align}
        \Pm_{X^n}(E) &= \int \mathbbm{1}_E \dd\Pm_{X^n} \\ 
        &=  \int \mathbbm{1}_E \frac{d\Pm_{X^n}}{d\Pm_{\bigotimes_{i=1}^n X_i}} \dd\Pm_{\bigotimes_{i=1}^n X_i} \\ &\leq \left(\int \mathbbm{1}_E \dd\Pm_{\bigotimes_{i=1}^n X_i}\right)^{\frac{\alpha-1}{\alpha}} \left(\int \left(\frac{d\Pm_{X^n}}{d\Pm_{\bigotimes_{i=1}^n X_i}}\right)^\alpha \dd\Pm_{\bigotimes_{i=1}^n X_i}\right)^\frac1\alpha  \\
        &= \Pm_{\bigotimes_{i=1}^n X_i}^{\frac{1}{\beta}}(E) H_\alpha^{\frac1\alpha}(\Pm_{X^n}\|\Pm_{\bigotimes_{i=1}^n X_i}),
\end{align}
where the inequality in the third line follows from H\"older's inequality. Moreover, 
\begin{align} \label{eq:startTensor}
H_\alpha(\Pm_{X^n}\|\Pm_{\bigotimes_{i=1}^n X_i})  &=  \Pm_{X_1}\left(\Pm_{\bigotimes_{i=2}^n X_i} \left(\frac{d\Pm_{X_2^{n}|X_1}}{d\Pm_{\bigotimes_{i=2}^n X_i}}\right)^\alpha\right) \\
&=\Pm_{X_1}\Bigg(\Pm_{X_2}\left(\left(\frac{d\Pm_{X_2|X_1}}{d\Pm_{X_2}}\right)^\alpha\left(\Pm_{\bigotimes_{i=3}^n X_i} \left(\frac{d\Pm_{X_3^{n}|X_2}}{d\Pm_{\bigotimes_{i=3}^n X_i}}\right)^\alpha\Bigg)\right)\right) \\
&\leq \Pm_{X_1}\left(\Pm_{X_2}\left(\left(\frac{d\Pm_{X_2|X_1}}{d\Pm_{X_2}}\right)^{\alpha\alpha_2}\right)^\frac{1}{\alpha_2}\Pm_{X_2}\left(\left(\Pm_{\bigotimes_{i=3}^n X_i} \left(\frac{d\Pm_{X_3^{n}|X_2}}{d\Pm_{\bigotimes_{i=3}^n X_i}}\right)^\alpha\right)^{\beta_2}\right)^\frac{1}{\beta_2}\right)  \\
&= H_\alpha^2 \cdot \Pm_{X_2}\left(\left(\Pm_{\bigotimes_{i=3}^n X_i} \left(\frac{d\Pm_{X_3^{n}|X_2}}{d\Pm_{\bigotimes_{i=3}^n X_i}}\right)^\alpha\right)^{\beta_2}\right)^\frac{1}{\beta_2},\label{eq:endTensor}
    \end{align}
where the inequality follows from H\"older's inequality, the fact that the expectation is taken with respect to the product measure $\Pm_{\bigotimes_i X_i}$ as well as the Markovianity of $\Pm_{X^n}$. $H_\alpha^2 = \Pm_{X_{1}}^\frac{1}{\beta_{1}}\left(H_{\alpha\alpha_2}^\frac{\beta_{1}}{\alpha_2}(\Pm_{X_2|X_{1}}\|\Pm_{X_2})\right)$ with $\beta_1=1$ will now be the only term depending on $X_1$. Repeating the same sequence of steps (an application of the Disintegration Theorem~\cite[Theorem 5.3.1]{ambrosioGradientFlowsMetric2008} to a Markovian setting, like in~\cite[Proposition 2.1]{gotzeWeaklyDependent}, followed by H\"older's inequality) $(n-2)$ more times leads to the product of $H_\alpha^i$ as defined in the statement of the theorem.
The result then follows by noticing that $$\Pm^\frac{1}{\beta}_{\bigotimes_i X_i}(E)\leq 2^\frac1\beta \exp\left(-\frac{2nt^2}{\beta} \right),$$ by McDiarmid's inequality.  
\end{proof}

\section{Concentration around mean and median}\label{app:cmm}
  The following result is a useful tool that allows us to compare concentration around a constant, concentration around the mean and concentration around the median. It is used in order to compare our results with the ones proposed in~\cite{Marton1996BoundingB,Marton1996}. A proof can be found in~\cite[Proposition 1.8] {ledoux2001concentration} and the statement is provided here for ease of reference.
    \begin{proposition}\label{prop:ledouxFromTailsToMedian}Let $f$ be a measurable function on a probability space $(\X,\Omega,\mu)$. Assume that, for some $a_f\in\mathbb{R}$ and a non-negative function $h:\mathbb{R}_+\to\mathbb{R}_+$ such that $
    \lim_{r\to\infty} h(r)=0$,
    \begin{equation}\mu(\left\{|f-a_f|\geq r \right\})\leq h(r)\end{equation} for all $r>0$, then
    \begin{equation}
        \mu(\left\{|f-m_f|\geq r +r_0\right\})\leq h(r),
    \end{equation}
    where $m_f$ is the median of $f$ and $r_0$ is such that $h(r_0)<\frac12$. Moreover, if $\bar{h} = \int_0^\infty h(x)dx < \infty$, then $f$ is $\mu$-integrable, $|a_f-\mu(f)|\leq\bar{h}$ and, for every $r>0$, 
    \begin{equation}
    \mu(\left\{|f-\mu(f)|\geq r +\bar{h}\right\})\leq h(r).
    \end{equation}
    In particular, if $h(r)\leq C\exp(-cr^p)$ with $0<p<\infty$, then
    \begin{equation}
    \mu(\left\{|f-M|\geq r \right\})\leq C'\exp(-\kappa_pcr^p),
    \end{equation}
    where $C'$ depends only on $C$ and $p$, $\kappa_p$ depends only on $p$, and $M$ is either the mean or the median.
    \end{proposition}
    \section{Proof of~\Cref{lemma:differenceOfMeans}}\label{app:boundOnMeans}
    \begin{proof}
        From~\Cref{thm:main} (see~\Cref{eq:thresholdEta}) we know that if  $t>t_\alpha$ one has an exponential decay in the probability, while if $t\leq t_\alpha$ then the exponent is positive and the bound is typically larger than $1$ and, thus, trivial. From this, we can prove that
\begin{align*}
    |\Pm_{X^n}(f) - \Pm_{\otimes_{i=1}^n X_i}(f)| &= \left|{\Pm_{X^n}}(f-c) \right| \\
    &\leq \Pm_{X^n}\left(|f-c|\right) \\
    &= \int_{0}^\infty \Pm_{X^n}(|f - c|\geq t) \dd t \\
    &\leq  \int_{0}^{t_\alpha} 1 \dd t + \int_{t_\alpha}^\infty \Pm_{X^n}(|f - c|\geq t) \dd t \\
    &\leq t_\alpha + 2^\frac1\beta H_\alpha^\frac1\alpha(\Pm_{X^n}\|\Pm_{\otimes_{i=1}^n X_i}) \int_{t_\alpha}^\infty \exp\left(-\frac{2n t^2}{\beta}\right) \dd t\\
    &\leq t_\alpha + 2^\frac1\beta H_\alpha^\frac1\alpha(\Pm_{X^n}\|\Pm_{\otimes_{i=1}^n X_i}) \frac{\beta}{2n t_\alpha} \int_{t_\alpha}^\infty \frac{2n t}{\beta}\exp\left(-\frac{2n t^2}{\beta}\right) \dd t \\
    &= t_\alpha + 2^\frac1\beta H_\alpha^\frac1\alpha(\Pm_{X^n}\|\Pm_{\otimes_{i=1}^n X_i}) \frac{\beta}{2n t_\alpha} \left(-\frac12 \exp\left(-\frac{2n t^2}{\beta}\right)\right)\Bigg\rvert_{t_\alpha}^{\infty}\\ 
    &= t_\alpha + 2^\frac1\beta H_\alpha^\frac1\alpha(\Pm_{X^n}\|\Pm_{\otimes_{i=1}^n X_i}) \frac{\beta}{4n t_\alpha} \exp\left(-\frac{2n t_\alpha^2}{\beta}\right) \\
    &= t_\alpha + 2^\frac1\beta\frac{\beta}{4n t_\alpha} \\
    &= t_\alpha + 2^\frac1\beta \frac{\beta}{4n} \frac{\sqrt{2n}}{\sqrt{\frac\beta\alpha \ln(H_\alpha(\Pm_{X^n}\|\Pm_{\otimes_{i=1}^n X_i}))}} \\
    &= t_\alpha +  \frac{\sqrt{\beta}}{2^\frac1\alpha\sqrt{\frac{2n}{\alpha} \ln(H_\alpha(\Pm_{X^n}\|\Pm_{\otimes_{i=1}^n X_i}))}},
\end{align*}
which gives the desired result.
    \end{proof}
\section{Tensorisation}\label{app:tensorisation}
      Many information measures satisfy tensorisation properties, meaning that, if $\nu$ and $\mu$ are measures acting on an $n$-dimensional space (typically an $n$-fold Cartesian product of one-dimensional spaces), it is possible to compute the divergence between $\nu$ and $\mu$ using ``one-dimensional projections''. This is particularly useful when the second measure is a product-measure. Indeed, if $\Pm$ and $\Q$ are two probability measures on the space $\X^n$ and $\Pm$ is a product-measure, denoting with $\bar{X}^i$ the $(n-1)$-tuple $(X_1,\ldots,X_{i-1},X_{i+1},\ldots,X_n)$, then~\cite[Proposition 3.1.2]{concentrationMeasureII} gives that
    \begin{equation}
        D(\Q\|\Pm) \leq \sum_{i=1}^n \int \dd\Q_{\bar{X}^i} D(Q_{X_i|\bar{X}^i}\|P_{X_i}),
    \end{equation}
    where $D(\Q\|\Pm)$ denotes the KL-divergence between $\Q$ and $\Pm$. 
    Hence, having access to a bound on $D(Q_{X_i|\bar{X}^i}\|P_{X_i})$ for every $1\leq i\leq n$ implies a bound on the KL-divergence between the two $n$-dimensional measures $\Q$ and $\Pm$. This property is pivotal in proving concentration results in a variety of ways~\cite[Section 3.1.3]{concentrationMeasureII}. Since of independent interest, we will now state the tensorisation properties of $H_\alpha$ as explicit results, as well as the corresponding R\'enyi's $D_\alpha$ analogues. 
    In particular, whenever the second measure is a product measure while the first one is Markovian, it is possible to prove the following.
    \begin{theorem}\label{thm:tensorisationMarkov}
    Let $\Q$ and $\Pm$ be two probability measures on the space $\X^n$ such that $\Q\ll\Pm$, and assume that $\Pm$ is a product measure (\textit{i.e.}, $\Pm=\bigotimes_{i=1}^n P_{i}$). Assume also that $\Q$ is a Markov measure induced by $Q_1$ and the kernels $Q_i(\cdot|\cdot)$ with $1\leq i\leq n$, \textit{i.e.}, $\Q(x^n) = Q_1(x_1)\prod_{i=2}^n Q_i(x_i|x_{i-1})$. Moreover, given a constant $c$, let $X_0=c$ (almost surely) be an auxiliary random variable, then, 
    \begin{equation}
        H_\alpha(Q\|\Pm) \leq
        \prod_{i=1}^n\Pm_{X_{i-1}}^\frac{1}{\beta_{i-1}}\left(H_{\alpha\alpha_i}^\frac{\beta_{i-1}}{\alpha_i}(Q_i(\cdot|X_{i-1})\|P_{X_i})\right), \label{eq:tensorisationGeneral}
    \end{equation}
    where $\alpha_i\geq 1$ for $i\geq 0$, $\beta_0=1$, $\alpha_n=1$, and $\beta_i = \alpha_i/(\alpha_i-1)$. 
    Moreover, selecting $\alpha_i\to 1^+$ which implies $\beta_{i}\to\infty$ for every $i\geq 1$, one recovers 
    \begin{equation}
        H_\alpha(\Q\|\Pm) \leq H_\alpha(Q_1\|P_1)\cdot\prod_{i=2}^n \max_{x_{i-1}\in \X} H_\alpha(Q_i(\cdot|x_{i-1})\| P_i). \label{eq:tensorisationMarkov}
    \end{equation}
    \end{theorem}
    \begin{proof}
        The proof follows from the steps undertaken in~\Crefrange{eq:startTensor}{eq:endTensor} along with the discussion immediately after, but replacing $\Pm_{X^n}$ with $\mathcal{Q}$. 
    \end{proof}
    ~\Cref{thm:tensorisationMarkov}, which involves products of Hellinger integrals, can be re-written as a sum by considering R\'enyi's divergences, due to the relationship connecting the two quantities (see~\Cref{eq:hellignerRenyi}). 
    \begin{corollary}
    Under the same assumptions as in~\Cref{thm:tensorisationMarkov}, one has that
    \begin{equation}
        D_\alpha(\Q\|\Pm) \leq \frac{1}{\alpha-1}\sum_{i=1}^n \frac{1}{\beta_{i-1}} \log \Pm_{X_{i-1}}\left(\exp\left(\frac{(\alpha\alpha_i-1)\beta_{i-1}}{\alpha_i}\left(D_{\alpha\alpha_i}(Q_{i}(\cdot|X_{i-1})\|P_{X_i}\right)\right)\right),
    \end{equation}
     where $\alpha_i\geq 1$ for $i\geq 0$, $\beta_0=1$, $\alpha_n=1$ and $\beta_i = \alpha_i/(\alpha_i-1)$. 
    Moreover, selecting $\alpha_i\to 1^+$ which implies $\beta_{i}\to\infty$ for every $i\geq 1$, one recovers 
        \begin{equation}
        D_\alpha(\Q\|\Pm) \leq D_\alpha(Q_1\|P_1) + \sum_{i=2}^n \max_{x_{i-1}} D_\alpha(Q_i(\cdot|x_{i-1})\| P_i). \label{eq:tensorisationRenyi}
    \end{equation}
    \end{corollary}

    This result allows us to control from above the Hellinger integral of two $n$-dimensional objects using the Hellinger integral of simpler $1$-dimensional objects.
    For instance, if the first measure represents the distribution induced by a time-homogeneous Markov chain, then all the kernels $Q_i$ coincide and the expression becomes even easier to compute, as one can see in~\Cref{sec:examples}.
    
    A similar approach can also be employed to provide a result in cases where $\Q$ is an arbitrary measure.
     \begin{theorem}
    Let $\Q$ and $\Pm$ be two probability measures on the space $\X^n$ such that $\Q\ll\Pm$, and assume that $\Pm$ is a product-measure (\textit{i.e.}, $\Pm=\bigotimes_{i=1}^n P_{i}$). Then, 
    \begin{equation}
        H_\alpha(\Q\|\Pm) \leq H_\alpha(Q_1\|P_1)\cdot\prod_{i=2}^n \max_{\textbf{x}_1^{i-1}=x_1\ldots x_{i-1}\in \X^{i-1}} H_\alpha(Q_i(\cdot|\textbf{x}_1^{i-1})\| P_i).
    \end{equation}
    Similarly, one has that 
     \begin{equation}
        D_\alpha(\Q\|\Pm) \leq D_\alpha(Q_1\|P_1) + \sum_{i=2}^n \max_{\textbf{x}_1^{i-1}=x_1\ldots x_{i-1}\in\mathcal{X}^{i-1} } D_\alpha(Q_i(\cdot|\textbf{x}_1^{i-1})\| P_i).
    \end{equation}
    \end{theorem}
    The proof follows from the same argument of~\Cref{thm:tensorisationMarkov}. The key difference is that, without making any additional assumptions on $\Q$, the entire ``past'' of the process needs to be considered. 
    This is why writing an expression similar to~\Cref{eq:tensorisationGeneral} without Markovianity would be rather cumbersome, and we restrict ourselves to the setting where the additional parameters are all considered to be such that  $\beta_{j}\to\infty$. 
    
    Finally, we show some lower bounds on Hellinger integrals and R\'enyi's divergences. 
    \begin{theorem}
        Let $\Q$ and $\Pm$ be two probability measures on the space $\X^n$ such that $\Q\ll\Pm$, and assume that $\Pm$ is a product-measure (\textit{i.e.}, $\Pm=\bigotimes_{i=1}^n P_{i}$). Assume also that $\Q$ is a Markov measure induced by $Q_1$ and the kernels $Q_i(\cdot|\cdot)$ with $1\leq i\leq n$, \textit{i.e.}, $\Q(x^n) = Q_1(x_1)\prod_{i=2}^n Q_i(x_i|x_{i-1})$. Moreover, given a constant $c$, let $X_0=c$ (almost surely) be an auxiliary random variable, then, 
    \begin{equation}
        H_\alpha(Q\|\Pm) \geq
        \prod_{i=1}^n\Pm_{X_{i-1}}^\frac{1}{\beta_{i-1}}\left(H_{\alpha\alpha_i}^\frac{\beta_{i-1}}{\alpha_i}(Q_i(\cdot|X_{i-1})\|P_{X_i})\right), 
    \end{equation}
    and
    \begin{equation}
        D_\alpha(\Q\|\Pm) \geq \frac{1}{\alpha-1}\sum_{i=1}^n \frac{1}{\beta_{i-1}} \log \Pm_{X_{i-1}}\left(\exp\left(\sign{(
        \alpha_i)}\frac{(\alpha\alpha_i-1)\beta_{i-1}}{\alpha_i}\left(D_{\alpha\alpha_i}(Q_{i}(\cdot|X_{i-1})\|P_{X_i}\right)\right)\right),
    \end{equation}
   where $\alpha_i\leq 1$ for $i\geq 1$, $\beta_0=1$, $\alpha_n=1$ and $\beta_i = \alpha_i/(\alpha_i-1)$. 
    Moreover, selecting $\alpha_1\to1^-$ which implies $\beta_{i}\to-\infty$ for every $i\geq 1$, one recovers 
    \begin{equation}
        H_\alpha(\Q\|\Pm) \geq H_\alpha(Q_1\|P_1)\cdot\prod_{i=2}^n \min_{x_{i-1}\in \X} H_\alpha(Q_i(\cdot|x_{i-1})\| P_i),\label{eq:tensorisationMarkovReverse}
    \end{equation}
    which, in the case of R\'enyi's divergences, specialises to
    \begin{equation}
        D_\alpha(\Q\|\Pm) \geq D_\alpha(Q_1\|P_1) + \sum_{i=2}^n \min_{x_{i-1}\in \mathcal{X}} D_\alpha(Q_i(\cdot|x_{i-1})\| P_i).
    \end{equation}
        \end{theorem}
    
    \begin{proof}
        The proof follows from similar arguments as  the proof of~\Cref{thm:tensorisationMarkov}. The sole difference is that, instead of using H\"older's inequality at each step, one uses reverse H\"older's inequality. 
    \end{proof}
\section{Explicit comparison for a binary kernel}\label{app:ex1}
The setting is the following: let $K$ be a time-homogeneous Markov chain  characterised by a doubly-stochastic transition matrix characterised by the vector of parameters $\bar{\lambda}= (\lambda_1,\ldots,\lambda_m)$ (\textit{i.e.}, the rows and columns of $K$ are permutations of $\bar{\lambda}$ with the constraints  $\sum_i K_{i,j} = \sum_j K_{i,j} = 1$). In this case, the Markov chain admits a stationary distribution $\pi$ which corresponds to the uniform distribution over the sample space. Hence, if one is considering an $m$-dimensional space, then $\pi(\{x\})=\frac1m$ for $x\in\{1,\ldots,m\}$.
    In this case, if $P_1 \sim \pi$, then $P_i \sim \pi$ for every $i\geq 1$. Moreover, one can prove that $K=K^\leftarrow$ and the bound of~\Cref{thm:main} reduces to \begin{equation}
        \Pm_{X^n}\left(\{| f-\Pm_{\bigotimes_n X_n}(f)|\geq t \right) \leq 2^\frac1\beta \exp\left(-\frac{1}{\beta}\left(2nt^2- (n-1)\log\left(m \left\lVert \bar{\lambda}\right\rVert_{\ell^\alpha}^\beta \right)\right)\right),\label{eq:probBoundGeneralChannel}
    \end{equation}
    Henceforth, we will consider $m=2$ for simplicity and, thus, $\left\lVert \bar{\lambda}\right\rVert_{\ell^\alpha}:=\left(\sum_{i=1}^m\lambda_i^\alpha\right)^\frac1\alpha$ can be expressed as $(\lambda^\alpha + (1-\lambda)^\alpha)^\frac1\alpha.$ Specialising~\Cref{eq:probBoundGeneralChannel} to this setting, one obtains the following result.
    
    \begin{corollary}\label{thm:contractingMarkov}
      Let $n>1$, and let $X_1,\ldots,X_n$ be a sequence of random variables such that $X_1\sim \pi$. For every $\alpha>1$ and every function $f$ such that~\Cref{eq:mcDiarmidAssumption} holds true, one has that  
      \begin{equation}      
          \mathbb{P}\left(\left\lvert f-\Pm_{\bigotimes_{i=1}^n X_i}(f)\right\rvert \geq t \right)\leq 2^\frac1\beta\exp\left(-\frac{2nt^2}{\beta} +\frac{n-1}{\beta}\ln\left(2((1-\lambda)^\alpha+\lambda^\alpha)^\frac{1}{\alpha-1}\right)\right) \label{eq:ourSimpleChannel}.
        \end{equation}
      \end{corollary}
\subsection{Comparison with~\cite[Theorem 1.2]{dependentViaMartingale}}\label{app:ex1Comp1}The bound obtained via the techniques of~\cite{dependentViaMartingale} is
      \begin{align}
          \mathbb{P}(\left\lvert f-\Pm_{X^n}(f)\right\rvert \geq t)&\leq 
          2\exp\left(-\frac{2\lambda^2 nt^2}{ \left((1-2\lambda)^n-1\right)^2}\right).\label{eq:kontorovichSimpleChannel}
      \end{align}
      Let us denote $\kappa_\alpha := ((1-\lambda)^\alpha+\lambda^\alpha))^\frac{1}{\alpha-1} <1$. Then, by direct comparison, it is possible to see that, whenever \begin{equation}t^2>\frac{((1-2\lambda)^n-1)^2(1-1/n)\ln(2\kappa_\alpha)}{2((1-2\lambda)^n-1)^2-\beta\lambda^2)}:=\bar{t}^2,\end{equation} 
    then the RHS of~\eqref{eq:ourSimpleChannel} decays faster than the RHS of~\eqref{eq:kontorovichSimpleChannel}. In particular, for a given $\lambda<\frac12$ and if  $\alpha>\frac43$, then
    \begin{equation}
        \bar{t}^2= (1+o_n(1))\frac{\ln(2\kappa_\alpha)}{2(1-\beta\lambda^2)} < (1+o_n(1))\frac{2\ln2}{4-\beta}. \label{eq:minimumEtaSimpleChannelBeta}
    \end{equation}


   Here, one can explicitly see the trade-off between the probability term and the Hellinger integral, described in~\Cref{rmk:tradeOff} and mediated by the choice of $\alpha$. Taking the limit $\alpha\to\infty$ in \Cref{eq:ourSimpleChannel} leads to the following upper bound
    \begin{equation}
    \mathbb{P}\left(\left\lvert f-\Pm_{\bigotimes_{i=1}^n X_i}(f)\right\rvert \geq t \right)\leq 2\exp\left(-2nt^2 + (n-1)\ln(2(1-\lambda))\right) \label{eq:ourSimpleChannel2}.
    \end{equation}
    In this setting, in order to improve over~\Cref{eq:kontorovichSimpleChannel}, one needs $t^2>\bar{t}^2$, with
    \begin{equation}
        \bar{t}^2 = (1+o_n(1))\frac{\ln(2(1-\lambda))}{2(1-\lambda^2)}< (1+o_n(1))\frac{2\ln2}{3}.\label{eq:minimumEtaSimpleChannelInfinity}
    \end{equation}
    Clearly, $1-\lambda>\kappa_\alpha$ for every $\alpha>1$. Hence, on the one hand, \Cref{eq:ourSimpleChannel2} introduces a worse multiplicative constant (a larger $\alpha$ implies a larger Hellinger integral, and we are considering the limit of $\alpha\to\infty$) and increases the minimum value of $t$ one can consider for a given $\lambda<1/2$ (\Cref{eq:minimumEtaSimpleChannelBeta} is increasing in $\alpha$). On the other hand, it provides a faster exponential decay with $n$. In fact, as $\beta\to 1$, the RHS of \eqref{eq:ourSimpleChannel2} scales as $\exp(-2nt^2(1+o_n(1)))$ for large enough $t$, which matches the behavior of~\Cref{eq:mcDiarmidStatement}. 
    \subsection{Comparison with~\cite{dependentViaStatAndChange}}\label{app:ex1Comp2}
    In the setting considered above, \cite[Theorem 1]{dependentViaStatAndChange} gives
    \begin{equation}
    \mathbb{P}\left(\left\lvert f-\pi^{\otimes n}(f)\right\rvert \geq t \right) \leq 2\exp\left(-\frac{2\lambda}{1-\lambda}nt^2\right). \label{eq:fanSimpleChannel}\end{equation}
    This means that, considering the decay provided by~\Cref{eq:ourSimpleChannel2} (which optimises the speed of decay for large enough $t$) whenever \begin{equation}
    t^2 > \frac{1-\lambda}{2-4\lambda}  \ln(2(1-\lambda))(1+o_n(1)),
    \label{eq:etaSimpleChannelFan2}
    \end{equation}
    then~\Cref{eq:ourSimpleChannel2} leads to a faster decay than~\Cref{eq:fanSimpleChannel}. 
    \subsection{Comparison with~\cite{Marton1996BoundingB}} \label{app:ex1Comp3}
    For the kernel considered in this appendix, \Cref{eq:convergenceMedianGeneral} holds with $C_n= (n-1)\log(2(1-\lambda))$. Furthermore, for every $i$ and $x, \hat{x}$, we have that
    $TV(\Pm_{X_i|X_{i-1}=x},\Pm_{X_i|X_{i-1}=\hat{x}})=(1-2\lambda).$ Consequently, assuming $t>1/(2\lambda)\sqrt{\ln(1/\Pm_{X^n}(E))/n}$, \cite[Proposition 1 \& Proposition 4]{Marton1996BoundingB} give
    \begin{align}            
    \Pm_{X^n}(E^c_t) &\leq \exp\left(-2n\left(t(2\lambda) -\sqrt{\frac{\ln(1/\Pm_{X^n}(E))}{2n}}\right)^2\right) \\
    & \leq \exp\left(-8nt^2\lambda^2+8t\lambda\sqrt{\frac{n\ln2}{2}}\right)\label{eq:martonGeometricSimpleChannel}.
    \end{align}
    Comparing~\Cref{eq:ourGeometricSimpleChannel} with $C_n=(n-1)\log(2(1-\lambda))$ with~\Cref{eq:martonGeometricSimpleChannel}, one can see that, whenever
     \begin{equation}
    t \geq \frac{\sqrt{2\ln(2(1-\lambda))}}{(1-4\lambda^2)}
\left(1+o_{n}(1)\right),
\end{equation}
then~\Cref{eq:ourGeometricSimpleChannel} improves over~\Cref{eq:martonGeometricSimpleChannel}. 

A similar comparison can be drawn with respect to the tools in~\cite{Marton1996} (in which the degree of dependence is measured differently), but it would lead to a worse bound than that expressed in~\Cref{eq:martonGeometricSimpleChannel}.
\subsection{MCMC}\label{app:MCMCBinary}
Considering the same setting detailed in the previous subsections, one can more explicitly characterise the parameters determining~\Cref{eq:thresholdEtaMCMC}. In particular, one has that $\pi = (1/2,1/2)$ and the spectral gap is equal to $1-2\lambda$. Consequently, if $n_0=0$ and one considers $\alpha\to\infty$,~\Cref{eq:MCMCFan} becomes:
\begin{equation}
    \mathbb{P}\left(\frac1n \sum_{i=1}^n f(X_{n_0+i})-\pi(f) > t \right) \leq 2\exp\left(-\frac{\lambda}{1-\lambda}\cdot \frac{2nt^2}{(b-a)^2} \right)\cdot \max\{ \nu(\{0\}),\nu(\{1\})\}, \label{eq:MCMCFan2Binary}
\end{equation}
while~\Cref{eq:MCMCOurn0} boils down to the following:
\begin{equation}
     \mathbb{P}\left(\frac1n \sum_{i=1}^n f(X_{n_0+i})-\pi(f) > t \right) \leq 2\exp\left(-\frac{2nt^2}{(b-a)^2}\right)(2-2\lambda)^{n-1}\cdot\max\{ \nu(\{0\}),\nu(\{1\})\}.\label{eq:MCMCOur2Binary}
\end{equation}
Making a direct comparison one can see that if
\begin{equation}
t^2 \geq \left(\left(1-\frac1n\right)\ln(2-2\lambda)\right)\frac{(b-a)^2}{2}\frac{1-\lambda}{1-2\lambda},
\end{equation}
then~\Cref{eq:MCMCOur2Binary} provides a faster decay than~\Cref{eq:MCMCFan2Binary}.
\subsection{Comparison between SDPI for $D_\alpha$
and hypercontractivity}\label{app:comparisonHyperContrDalpha}
If $\mu=(1/2,1/2)$ and $K=\text{BSC}(\lambda)$, then $\mu K= \mu$ and one has, following Wyner's notation~\cite[Eq. 1.17]{wynerCommon}, the so-called Doubly-Symmetric Binary Source ``DSBS$(\lambda)$''. In this setting, the hyper-contractivity constant is given by \cite{hypercontractivity,hypercontractivityBonamiBeckner} \begin{equation}
 \gamma^\star(\alpha)=(1-2\lambda)^2(\alpha-1)-1 \label{eq:hypercontrDSBS}.
\end{equation}Moreover, one can analytically see that, for every $\mu$,
\begin{equation} \frac{D_\alpha(K\| \mu K)}{D_\alpha(\delta_0\| \mu)} < \frac{\left(1-2\lambda\right)^{\left(1+\frac{1}{\alpha}\right)}}{\left(1-\lambda\right)^{\frac{\left(\alpha-1\right)}{\alpha}}}.\label{eq:SDPIDSBSRenyi}
\end{equation}
As $\alpha\to1^+$, the LHS of~\Cref{eq:SDPIDSBSRenyi} approaches a ratio between KL-divergences, while the RHS approaches $(1-2\lambda)^2$, which equals $\eta_{KL}(K)$~\cite{sdpiRaginsky}. Furthermore, taking the limit of $\alpha\to\infty$ (which optimises the exponential rate of decay), the expression in~\Cref{eq:SDPIDSBSRenyi} provides an improvement over simply using DPI, while~\Cref{eq:hypercontrDSBS} trivialises to $\gamma^\star(+\infty) = +\infty$. Hence, denoting with $E=\{|f-\Pm_{\bigotimes X_i}(f)|\geq t\}$, one has that, 
for every $\alpha>1$,
\begin{align}
    \Pm_{X^n}(E) \leq 2^\frac1\beta \exp\left(-\frac{2nt^2}{\beta}  \right)\cdot\begin{cases}
           \exp\left(\frac{(1-2\lambda)^2(\alpha-1)-2}{(1-2\lambda)^2(\alpha-1)-1}(n-1)\log2\right)& \text{via \Cref{eq:generalResultDiscreteHomogeneous2Hyper} \& \Cref{eq:hypercontrDSBS}}, \\ \exp\left(\frac{1}{\beta} \frac{\left(1-2\lambda\right)^{\left(1+\frac{1}{\alpha}\right)}}{\left(1-\lambda\right)^{\frac{\left(\alpha-1\right)}{\alpha}}} (n-1)\log2\right) & \text{via \Cref{eq:SDPIDalphaBound} \& \Cref{eq:SDPIDSBSRenyi}}.
    \end{cases} \label{eq:comparisonDPIHyper}
\end{align}
One can see that, for $\alpha$ large enough, \Cref{eq:SDPIDalphaBound} improves upon~\Cref{eq:generalResultDiscreteHomogeneous2Hyper} for every $\lambda$. In fact, taking $\alpha\to\infty$ gives
\begin{equation}
    \Pm_{X^n}(E) \leq 2 \exp\left(-2nt^2  \right)\cdot\begin{cases}
           \exp\left((n-1)\log2\right)& \text{via \Cref{eq:generalResultDiscreteHomogeneous2Hyper} \& \Cref{eq:hypercontrDSBS}}, \\ \exp\left(\frac{\left(1-2\lambda\right)}{\left(1-\lambda\right)} (n-1)\log2\right) & \text{via \Cref{eq:SDPIDalphaBound} \& \Cref{eq:SDPIDSBSRenyi}}.
    \end{cases} 
\end{equation}
\section{Proof of \Cref{lemma:hAlphaSSRW}}\label{app:proofLemmaSSRW}

\begin{proof}
    For every $n\geq 0$, we have $$\text{supp}(S_n)=\bigcup_{j=0}^n \{ n-2j\}.$$ Moreover, given $0\leq j \leq n$, 
$$\mathbb{P}(S_n= n-2j) = \binom{n}{n-j}2^{-n}.$$ Furthermore, given $x\in \text{supp}(S_{n-1})$ and $0\leq j\leq n$,
    $$\mathbb{P}_{S_n|S_{n-1}=x}(n-2j)= \frac12\mathbbm{1}_{\{n-2j-x=1\}} + \frac12 \mathbbm{1}_{\{n-2j-x=-1\}}.$$
Therefore, the following chain of inequalities holds:
        \begin{align}
        H_\alpha (P_{S_n|S_{n-1}=x}\|P_{S_n}) &= \sum_{j=0}^n \mathbb{P}^\alpha_{S_n|S_{n-1}=x}(n-2j) \cdot \mathbb{P}_{S_n}^{1-\alpha}(n-2j) \notag\\ 
        &= 2^{-\alpha} \left(\mathbb{P}_{S_n}^{1-\alpha}(x+1)+\mathbb{P}_{S_n}^{1-\alpha}(x-1)\right)\notag \\
        &= 2^{-\alpha}\left(\left(\binom{n}{\frac{n+x+1}{2}} 2^{-n}\right)^{1-\alpha}+\left(\binom{n}{\frac{n+x-1}{2}} 2^{-n}\right)^{1-\alpha}\right)\notag\\
        &= 2^{-\alpha+n(\alpha-1)}\left(\binom{n}{\frac{n+x+1}{2}}^{1-\alpha}+\binom{n}{\frac{n+x-1}{2}} ^{1-\alpha}\right)\notag\\
        &\leq 2^{-\alpha+n(\alpha-1)}\left(\left(\frac{n}{\frac{n+x+1}{2}} \right)^{\frac{n+x+1}{2}(1-\alpha)}+\left(\frac{n}{\frac{n+x-1}{2}} \right)^{\frac{n+x-1}{2}(1-\alpha)}\right) \label{eq:firstIneqBinomial}\\
        &= 2^{-\alpha+n(\alpha-1)}\left(\left(\frac{n+x+1}{2n} \right)^{\frac{n+x+1}{2}(\alpha-1)}+\left(\frac{n+x-1}{2n} \right)^{\frac{n+x-1}{2}(\alpha-1)}\right)\\
        &\leq 2^{-\alpha+n(\alpha-1)}\left(\left(\frac{n+x+1}{2n} \right)^{\frac{n+x+1}{2}(\alpha-1)}+\left(\frac{n+x+1}{2n} \right)^{\frac{n+x-1}{2}(\alpha-1)}\right) \label{eq:secondInequalityBinom} \\
        &= 2^{-\alpha+n(\alpha-1)} \left(\frac{n+x+1}{2n}\right)^{\frac{n+x-1}{2}(\alpha-1)} \cdot\left(1+\left(\frac{n+x+1}{2n}\right)^{(\alpha-1)}\right) \label{eq:firstEqualityBinom} \\
        &\leq 
        2^{-\alpha+n(\alpha-1)+1}.\notag
        \end{align}      
        Here, the inequality \eqref{eq:firstIneqBinomial} follows from the fact that $\binom{n}{k}\geq \left(\frac{n}{k}\right)^k$ along with $1-\alpha\leq 0$; the inequality \eqref{eq:secondInequalityBinom} follows from the fact that $\frac{n+x+1}{2n}>0$. 
        Moreover, it is easy to see that  $\frac{n+x+1}{2}= \frac{n+x-1}{2}+1$ and thus the factorisation in~\Cref{eq:firstEqualityBinom} follows. To conclude, it suffice to notice that  $\frac{n+x+1}{2n}\leq 1$.
        Denoting  with $\frac1\beta = \frac{\alpha-1}{\alpha}$, the computations just above imply that
    \begin{align*}
    H_\alpha^\frac1\alpha(\Pm_{S^n}\|\Pm_{\bigotimes_{j=1}^n S_j}) &\leq \prod_{i=2}^n \max_{x\in \text{supp}{S_{i-1}}} H_\alpha^\frac1\alpha(\Pm_{S_i|S_{i-1=x}}\|\Pm_{S_i}) \\
    &\leq \prod_{i=2}^n 2^{\frac1\alpha-1 + \frac{i}{\beta}} =  \prod_{i=2}^n 2^{-\frac1\beta + \frac{i}{\beta}} = 2^{\frac1\beta\sum_{i=2}^n (i-1)} = 2^{\frac1\beta\sum_{j=1}^{n-1} j} = 2^{\frac{n(n-1)}{2\beta}}.
    \end{align*}
    For the lower bound we need the following observations. Let us denote with
    \begin{equation}
         h(k)= \binom{n}{\frac{n+k+1}{2}}^{1-\alpha}+\binom{n}{\frac{n+k-1}{2}} ^{1-\alpha},
    \end{equation}
    one can easily prove that:
    \begin{itemize}
        \item $h(k)=h(-k)$ \textit{i.e.}, it is even;
        \item $h(k)$ is decreasing if $k \in \{-n+1, \ldots, 0\}$ and increasing if $k\in \{0,\ldots, n-1\}$.
    \end{itemize}
   Hence, $\argmin_{k \in \{-n+1,\ldots,n-1\}} h(k) = 0$ and $\argmax_{k \in \{-n+1,\ldots,n-1\}} h(k) = \{n-1,-n+1\}$.
   One thus has that:
   \begin{align}
        H_\alpha (P_{S_n|S_{n-1}=x}\|P_{S_n}) &= 2^{-\alpha+n(\alpha-1)}\left(\binom{n}{\frac{n+x+1}{2}}^{1-\alpha}+\binom{n}{\frac{n+x-1}{2}} ^{1-\alpha}\right) \\
        &\geq 2^{-\alpha+n(\alpha-1)} 2 \binom{n}{\frac{n+1}{2}}^{1-\alpha} \label{eq:minimiser} \\
        &= 2^{-\alpha+n(\alpha-1)+1} \left(\frac{n!}{\left(\frac{n+1}{2}\right)!\left(\frac{n-1}{2}\right)!}\right)^{1-\alpha}
        \\ &\geq 2^{-\alpha+n(\alpha-1)+1} \left(\sqrt{\frac{n}{2\pi \left(\frac{n+1}{2} \frac{n-1}{2}\right)}}2^{nh_2\left(\frac{n+1}{2n}\right)}\right)^{1-\alpha}\label{eq:upperboundBinomial}
        \\ &= 2^{-\alpha+n(\alpha-1)+1-(\alpha-1) n h_2(\frac{n+1}{2n})} \left(\sqrt{\frac{2n}{\pi \left(n^2-1\right)}}\right)^{1-\alpha}
        \\
        &= 2^{-\alpha+n(\alpha-1)+1-(\alpha-1) n h_2(\frac{n+1}{2n})+ \frac{(\alpha-1)}{2} \log\left(\frac{\pi}{2}\frac{n^2-1}{n}\right)},
   \end{align}
   where~\Cref{eq:minimiser} follows from selecting $x=0$,~\Cref{eq:upperboundBinomial} follows from the fact that $$\binom{n}{k}\leq \sqrt{\frac{n}{2\pi k(n-k)}}2^{n h_2(k/n)},$$ where $h_2(x)=-x\log_2(x)-(1-x)\log_2(1-x)$ denotes the binary entropy (see~\cite[Problem 5.8]{informationTheoryGallager}). Thus,
   \begin{align}
        \log_2 H_\alpha^\frac1\alpha(\Pm_{S^n}\|\Pm_{\bigotimes_{j=1}^n S_j}) &\geq \log_2 \prod_{i=2}^n \min_{x\in \text{supp}{S_{i-1}}} H_\alpha^\frac1\alpha(\Pm_{S_i|S_{i-1=x}}\|\Pm_{S_i}) \\
    &= \frac1\alpha \left(\sum_{i=2}^n (1-\alpha)+i(\alpha-1)-(\alpha-1) i h_2\left(\frac{i+1}{2i}\right)+ \frac{(\alpha-1)}{2} \log_2\left(\frac{\pi}{2}\left(\frac{i^2-1}{i}\right)\right)\right) \\
    &= \sum_{i=2}^n -\frac1\beta +\frac1\beta i \left(1-h_2\left(\frac{i+1}{2i}\right)\right)+\frac{1}{2\beta}\log_2\left(\frac{\pi}{2}\left(\frac{i^2-1}{i}\right)\right)\\
    &= \frac1\beta\left( \sum_{i=2}^n -1 + i\left(1-h_2\left(\frac{i+1}{2i}\right)\right)+\frac{1}{2}\log_2\left(\frac{\pi}{2}\left(\frac{i^2-1}{i}\right)\right)\right)
    \\&\geq \frac1\beta\left(\frac14+ \sum_{i=4}^n -1 +i\left(1-h_2\left(\frac{i+1}{2i}\right)\right)+\frac{1}{2}\log_2\left(\frac{\pi}{2}\left(\frac{i^2-1}{i}\right)\right)\right)\label{eq:firstTwoTerms}
    \\
    &\geq \frac1\beta\left(\frac14 + \sum_{i=4}^n -1 + \frac{1}{2}\log_2\left(\frac{\pi}{2}\left(\frac{i^2-1}{i}\right)\right)\right) \label{eq:lowerBoundEntropy}\\
    &\geq \frac1\beta\left(\frac14 - (n-3) + \frac{(n-3)}{2}\log_2\left(\frac{15\pi}{8}\right)\right) \label{eq:lowerBoundLog}\\
    &= \frac1\beta\left(\frac14 +(n-3)\left(-1+\frac{1}{2}\log_2\left(\frac{15\pi}{8}\right)\right)\right) \\
    &\geq \frac1\beta\left(\frac14 +\frac{(n-3)}{4}\right) = \frac{1}{4\beta}(n-2),
   \end{align}
   where~\Cref{eq:firstTwoTerms} follows as the first two terms in the summation add up to $\approx 0.27$,~\Cref{eq:lowerBoundEntropy}  follows from the fact that $0\leq h_2(k) \leq 1$ with equality if and only if $k=\frac12$ and thus $1-h_2((i+1)/2i) \geq 0$ (for $i$ large enough the difference approaches $0$). Moreover, since $g(x)=(x^2-1)/x$ is an increasing function and thus $g(x)\geq g(4)$ for $x\geq 4$, one also has that \Cref{eq:lowerBoundLog} holds. To conclude, one can see that $\frac12\log_2\left(\frac{15\pi}{8}\right)-1>\frac{1}{4}$. 
   
\end{proof}
    \section{Proof of the inequalities in \eqref{eq:det1} and \eqref{eq:det2}}
    \label{app:computationsNonMarkovianEx}
    
Given the setting, denoting with $y_n = \sum_{i=1}^{n-1}p_ix_i$, one has that
\begin{align*}
H_\alpha(\Pm_{X^n|X^{n-1}=x_1^{n-1}}\|(1/2,1/2)) &= \left(\frac12\right)^{1-\alpha}\left(\left(\frac12+\sum_{i=1}^{n-1}p_ix_i\right)^\alpha+\left(\frac12-\sum_{i=1}^{n-1}p_ix_i\right)^\alpha\right) \\
 &\leq \left(\frac12\right)^{1-\alpha}\left(\left(\frac12+\sum_{i=1}^{n-1}p_ix_i\right)^{\lfloor\alpha\rfloor}+\left(\frac12-\sum_{i=1}^{n-1}p_ix_i\right)^{\lfloor\alpha\rfloor}\right) \\
    &= \left(\frac12\right)^{1-\alpha}\Bigg(\sum_{k=0}^{\lfloor\alpha\rfloor} \binom{{\lfloor\alpha\rfloor}}{k} \left(\frac12\right)^{{\lfloor\alpha\rfloor}-k}y_n^k + \sum_{k=0}^{\lfloor\alpha\rfloor} \binom{{\lfloor\alpha\rfloor}}{k} \left(\frac12\right)^{{\lfloor\alpha\rfloor}-k}(-y_n)^k\Bigg) \\
    &= \left(\frac12\right)^{1-\alpha}\left(\sum_{k=0}^{\lfloor\alpha\rfloor} \binom{{\lfloor\alpha\rfloor}}{k} \left(\frac12\right)^{{\lfloor\alpha\rfloor}-k}\left(y_n^k +(-y_n)^k\right)\right)
\end{align*}
\begin{align*}
&= \left(\frac12\right)^{1-\alpha}\left(\sum_{j=0}^{\lfloor\alpha/2\rfloor} \binom{{\lfloor\alpha\rfloor}}{2j} \left(\frac12\right)^{{\lfloor\alpha\rfloor}-2j}2y_n^{2j}\right) \\
    &\leq 2\sum_{j=0}^{\lfloor\alpha/2\rfloor} \binom{{\lfloor\alpha\rfloor}}{2j} (2y_n)^{2j}.
\end{align*}
Moreover, setting $p_i = 2^{-i-1}$ gives 
\begin{align*}
\max_{x_1^{n-1}}H_\alpha(\Pm_{X^n|X^{n-1}=x_1^{n-1}}\|(1/2,1/2)) &\leq2\max_{x_1^{n-1}}\sum_{j=0}^{\lfloor\alpha/2\rfloor} \binom{{\lfloor\alpha\rfloor}}{2j} \left(2\sum_{i=1}^{n-1}p_ix_i\right)^{2j} \\
 &=2 \sum_{j=0}^{\lfloor\alpha/2\rfloor} \binom{{\lfloor\alpha\rfloor}}{2j} \left(\sum_{i=1}^{n-1}2^{-i}\right)^{2j} \\
 &=2\sum_{j=0}^{\lfloor\alpha/2\rfloor} \binom{{\lfloor\alpha\rfloor}}{2j} \left(1-2^{-n+1}\right)^{2j} \\
 &\leq 2\sum_{j=0}^{\lfloor\alpha/2\rfloor} \binom{{\lfloor\alpha\rfloor}}{2j} = 2^{{\lfloor\alpha\rfloor}} \leq 2^{\alpha}.
\end{align*}

	%
	
	
	

\end{document}